\documentclass[journal]{IEEEtran}
\usepackage{amsmath,amsfonts}
\usepackage{caption}
\usepackage{algorithmic}
\usepackage{algorithm}
\usepackage{amsmath, amsthm}
\usepackage{array}
\usepackage{subcaption}
\usepackage{textcomp}
\usepackage{stfloats}
\usepackage{url}
\usepackage{verbatim}
\usepackage{graphicx}
\usepackage{amsthm}
\usepackage{cite,graphicx,amsmath,amssymb}
\usepackage{fancyhdr}
\usepackage{hhline}
\usepackage{graphicx,graphics}
\usepackage{array,color}
\usepackage{mathtools}
\usepackage{amsmath}
\usepackage[T1]{fontenc}
\usepackage{float}
\usepackage{enumitem}

\usepackage{cite}
\begin{document}
\newtheorem{theorem}{\emph{\underline{Theorem}}}
\newtheorem{acknowledgement}[theorem]{Acknowledgement}
\renewcommand{\algorithmicensure}{ \textbf{repeat:}}
\newtheorem{axiom}[theorem]{Axiom}
\newtheorem{case}[theorem]{Case}
\newtheorem{claim}[theorem]{Claim}
\newtheorem{conclusion}[theorem]{Conclusion}
\newtheorem{condition}[theorem]{Condition}
\newtheorem{conjecture}[theorem]{\emph{\underline{Conjecture}}}
\newtheorem{criterion}[theorem]{Criterion}
\newtheorem{definition}{\emph{\underline{Definition}}}
\newtheorem{exercise}[theorem]{Exercise}
\newtheorem{lemma}{\emph{\underline{Lemma}}}
\newtheorem{corollary}{\emph{\underline{Corollary}}}
\newtheorem{notation}[theorem]{Notation}
\newtheorem{problem}[theorem]{Problem}
\newtheorem{proposition}{\emph{\underline{Proposition}}}
\newtheorem{solution}[theorem]{Solution}
\newtheorem{summary}[theorem]{Summary}
\newtheorem{assumption}{Assumption}
\newtheorem{example}{\bf \emph{\underline{Example}}}
\newtheorem{remark}{\bf \emph{\underline{Remark}}}
\newtheorem{property}{\emph{\underline{Property}}}

\def\qed{$\Box$}
\def\QED{\mbox{\phantom{m}}\nolinebreak\hfill$\,\Box$}
\def\proof{\noindent{\emph{Proof:} }}
\def\poof{\noindent{\emph{Sketch of Proof:} }}
\def
\endproof{\hspace*{\fill}~\qed
\par
\endtrivlist\unskip}
\def\endproof{\hspace*{\fill}~\qed\par\endtrivlist\vskip3pt}

\def\E{\mathsf{E}}
\def\eps{\varepsilon}
\def\Lsp{{\boldsymbol L}}
\def\Bsp{{\boldsymbol B}}
\def\lsp{{\boldsymbol\ell}}
\def\Ltsp{{\Lsp^2}}
\def\Lpsp{{\Lsp^p}}
\def\Linsp{{\Lsp^{\infty}}}
\def\LtR{{\Lsp^2(\Rst)}}
\def\ltZ{{\lsp^2(\Zst)}}
\def\ltsp{{\lsp^2}}
\def\ltZt{{\lsp^2(\Zst^{2})}}
\def\ninN{{n{\in}\Nst}}
\def\oh{{\frac{1}{2}}}
\def\grass{{\cal G}}
\def\ord{{\cal O}}
\def\dist{{d_G}}
\def\conj#1{{\overline#1}}
\def\ntoinf{{n \rightarrow \infty }}
\def\toinf{{\rightarrow \infty }}
\def\tozero{{\rightarrow 0 }}
\def\trace{{\operatorname{trace}}}
\def\ord{{\cal O}}
\def\UU{{\cal U}}
\def\rank{{\operatorname{rank}}}
\def\acos{{\operatorname{acos}}}

\def\SINR{\mathsf{SINR}}
\def\SNR{\mathsf{SNR}}
\def\SIR{\mathsf{SIR}}
\def\tSIR{\widetilde{\mathsf{SIR}}}
\def\Ei{\mathsf{Ei}}
\def\l{\left}
\def\r{\right}
\def\({\left(}
\def\){\right)}
\def\lb{\left\{}
\def\rb{\right\}}

\setcounter{page}{1}

\newcommand{\eref}[1]{(\ref{#1})}
\newcommand{\fig}[1]{Fig.\ \ref{#1}}

\def\bydef{:=}
\def\ba{{\mathbf{a}}}
\def\bb{{\mathbf{b}}}
\def\bc{{\mathbf{c}}}
\def\bd{{\mathbf{d}}}
\def\bee{{\mathbf{e}}}
\def\bff{{\mathbf{f}}}
\def\bg{{\mathbf{g}}}
\def\bh{{\mathbf{h}}}
\def\bi{{\mathbf{i}}}
\def\bj{{\mathbf{j}}}
\def\bk{{\mathbf{k}}}
\def\bl{{\mathbf{l}}}
\def\bm{{\mathbf{m}}}
\def\bn{{\mathbf{n}}}
\def\bo{{\mathbf{o}}}
\def\bp{{\mathbf{p}}}
\def\bq{{\mathbf{q}}}
\def\br{{\mathbf{r}}}
\def\bs{{\mathbf{s}}}
\def\bt{{\mathbf{t}}}
\def\bu{{\mathbf{u}}}
\def\bv{{\mathbf{v}}}
\def\bw{{\mathbf{w}}}
\def\bx{{\mathbf{x}}}
\def\by{{\mathbf{y}}}
\def\bz{{\mathbf{z}}}
\def\b0{{\mathbf{0}}}

\def\bA{{\mathbf{A}}}
\def\bB{{\mathbf{B}}}
\def\bC{{\mathbf{C}}}
\def\bD{{\mathbf{D}}}
\def\bE{{\mathbf{E}}}
\def\bF{{\mathbf{F}}}
\def\bG{{\mathbf{G}}}
\def\bH{{\mathbf{H}}}
\def\bI{{\mathbf{I}}}
\def\bJ{{\mathbf{J}}}
\def\bK{{\mathbf{K}}}
\def\bL{{\mathbf{L}}}
\def\bM{{\mathbf{M}}}
\def\bN{{\mathbf{N}}}
\def\bO{{\mathbf{O}}}
\def\bP{{\mathbf{P}}}
\def\bQ{{\mathbf{Q}}}
\def\bR{{\mathbf{R}}}
\def\bS{{\mathbf{S}}}
\def\bT{{\mathbf{T}}}
\def\bU{{\mathbf{U}}}
\def\bV{{\mathbf{V}}}
\def\bW{{\mathbf{W}}}
\def\bX{{\mathbf{X}}}
\def\bY{{\mathbf{Y}}}
\def\bZ{{\mathbf{Z}}}

\def\mA{{\mathbb{A}}}
\def\mB{{\mathbb{B}}}
\def\mC{{\mathbb{C}}}
\def\mD{{\mathbb{D}}}
\def\mE{{\mathbb{E}}}
\def\mF{{\mathbb{F}}}
\def\mG{{\mathbb{G}}}
\def\mH{{\mathbb{H}}}
\def\mI{{\mathbb{I}}}
\def\mJ{{\mathbb{J}}}
\def\mK{{\mathbb{K}}}
\def\mL{{\mathbb{L}}}
\def\mM{{\mathbb{M}}}
\def\mN{{\mathbb{N}}}
\def\mO{{\mathbb{O}}}
\def\mP{{\mathbb{P}}}
\def\mQ{{\mathbb{Q}}}
\def\mR{{\mathbb{R}}}
\def\mS{{\mathbb{S}}}
\def\mT{{\mathbb{T}}}
\def\mU{{\mathbb{U}}}
\def\mV{{\mathbb{V}}}
\def\mW{{\mathbb{W}}}
\def\mX{{\mathbb{X}}}
\def\mY{{\mathbb{Y}}}
\def\mZ{{\mathbb{Z}}}

\def\cA{\mathcal{A}}
\def\cB{\mathcal{B}}
\def\cC{\mathcal{C}}
\def\cD{\mathcal{D}}
\def\cE{\mathcal{E}}
\def\cF{\mathcal{F}}
\def\cG{\mathcal{G}}
\def\cH{\mathcal{H}}
\def\cI{\mathcal{I}}
\def\cJ{\mathcal{J}}
\def\cK{\mathcal{K}}
\def\cL{\mathcal{L}}
\def\cM{\mathcal{M}}
\def\cN{\mathcal{N}}
\def\cO{\mathcal{O}}
\def\cP{\mathcal{P}}
\def\cQ{\mathcal{Q}}
\def\cR{\mathcal{R}}
\def\cS{\mathcal{S}}
\def\cT{\mathcal{T}}
\def\cU{\mathcal{U}}
\def\cV{\mathcal{V}}
\def\cW{\mathcal{W}}
\def\cX{\mathcal{X}}
\def\cY{\mathcal{Y}}
\def\cZ{\mathcal{Z}}
\def\cd{\mathcal{d}}
\def\Mt{M_{t}}
\def\Mr{M_{r}}
\def\O{\Omega_{M_{t}}}
\newcommand{\figref}[1]{{Fig.}~\ref{#1}}
\newcommand{\tabref}[1]{{Table}~\ref{#1}}

\newcommand{\var}{\mathsf{var}}
\newcommand{\fb}{\tx{fb}}
\newcommand{\nf}{\tx{nf}}
\newcommand{\BC}{\tx{(bc)}}
\newcommand{\MAC}{\tx{(mac)}}
\newcommand{\Pout}{p_{\mathsf{out}}}
\newcommand{\nnn}{\nn\\}
\newcommand{\FB}{\tx{FB}}
\newcommand{\TX}{\tx{TX}}
\newcommand{\RX}{\tx{RX}}
\renewcommand{\mod}{\tx{mod}}
\newcommand{\m}[1]{\mathbf{#1}}
\newcommand{\td}[1]{\tilde{#1}}
\newcommand{\sbf}[1]{\scriptsize{\textbf{#1}}}
\newcommand{\stxt}[1]{\scriptsize{\textrm{#1}}}
\newcommand{\suml}[2]{\sum\limits_{#1}^{#2}}
\newcommand{\sumlk}{\sum\limits_{k=0}^{K-1}}
\newcommand{\eqhsp}{\hspace{10 pt}}
\newcommand{\tx}[1]{\texttt{#1}}
\newcommand{\Hz}{\ \tx{Hz}}
\newcommand{\sinc}{\tx{sinc}}
\newcommand{\tr}{\mathrm{tr}}
\newcommand{\diag}{\mathrm{diag}}
\newcommand{\MAI}{\tx{MAI}}
\newcommand{\ISI}{\tx{ISI}}
\newcommand{\IBI}{\tx{IBI}}
\newcommand{\CN}{\tx{CN}}
\newcommand{\CP}{\tx{CP}}
\newcommand{\ZP}{\tx{ZP}}
\newcommand{\ZF}{\tx{ZF}}
\newcommand{\SP}{\tx{SP}}
\newcommand{\MMSE}{\tx{MMSE}}
\newcommand{\MINF}{\tx{MINF}}
\newcommand{\RC}{\tx{MP}}
\newcommand{\MBER}{\tx{MBER}}
\newcommand{\MSNR}{\tx{MSNR}}
\newcommand{\MCAP}{\tx{MCAP}}
\newcommand{\vol}{\tx{vol}}
\newcommand{\ah}{\hat{g}}
\newcommand{\tg}{\tilde{g}}
\newcommand{\teta}{\tilde{\eta}}
\newcommand{\heta}{\hat{\eta}}
\newcommand{\uh}{\m{\hat{s}}}
\newcommand{\eh}{\m{\hat{\eta}}}
\newcommand{\hv}{\m{h}}
\newcommand{\hh}{\m{\hat{h}}}
\newcommand{\Po}{P_{\mathrm{out}}}
\newcommand{\Poh}{\hat{P}_{\mathrm{out}}}
\newcommand{\Ph}{\hat{\gamma}}
\newcommand{\mat}[1]{\begin{matrix}#1\end{matrix}}
\newcommand{\ud}{^{\dagger}}
\newcommand{\C}{\mathcal{C}}
\newcommand{\nn}{\nonumber}
\newcommand{\nInf}{U\rightarrow \infty}

\title{Near-field Beam-focusing Pattern under \\ Discrete Phase Shifters}

\author{Haodong Zhang, Changsheng You,~\IEEEmembership{Member,~IEEE}, Cong Zhou
\thanks{Part of this work will be presented at the IEEE WCNC 2025, Milan, Italy, March 24–27, 2025 \cite{zhang2025near}.}
\thanks{Haodong Zhang and Changsheng You are with the Department of Electronic and Electrical Engineering, Southern University of Science and Technology, Shenzhen 518055, China. (e-mails: zhanghd2021@mail.sustech.edu.cn; youcs@sustech.edu.cn). }
\thanks{Cong Zhou is with the School of Electronic and Information Engineering, Harbin Institute of Technology, Harbin, 150001, China, and also with the Department of Electrical and Electronic Engineering, Southern University of Science and Technology, Shenzhen 518055, China  (e-mail:  zhoucong@stu.hit.edu.cn).}
\thanks{\emph{(Corresponding author: Changsheng You.)}}
}

\maketitle

\begin{abstract}
Extremely large-scale arrays (XL-arrays) have emerged as a promising technology for enabling \emph{near-field} communications in future wireless systems.
However, the huge number of antennas deployed pose demanding challenges on the hardware cost and power consumption, especially when the antennas employ high-resolution phase shifters (PSs). 
To address this issue, in this paper, we consider low-resolution \emph{discrete} PSs at the XL-array which are practically more energy efficient, and investigate the impact of PS resolution on the near-field \emph{beam-focusing} effect. 
To this end, we propose a new \emph{Fourier series expansion} method to efficiently tackle the difficulty in characterizing the beam pattern properties under phase quantization. 
Interestingly, we analytically show, for the first time, that 1) discrete PSs introduce additional \emph{grating} lobes; 2) the main lobe still exhibits the beam-focusing property with its beam power increasing with PS resolution; and 3) there are two types of grating lobes, featured by the beam-focusing and beam-steering properties, respectively. 
In addition, we provide intuitive understanding for the appearance of grating lobes under discrete PSs from an \emph{array-of-subarrays} perspective.
Finally, numerical results demonstrate that the grating lobes generally degrade communication rate performance. 
However, a low-resolution of 3-bit PSs can achieve similar beam pattern and rate performance with the continuous PS counterpart, while it attains much higher energy efficiency. 

\end{abstract}


\section{Introduction}
Extremely large-scale arrays (XL-arrays) have emerged as a promising technology to enhance the spectral efficiency and spatial resolution in future wireless systems~\cite{lu2024tutorial,you2024next, cui2022near,cong2024near}. In addition, via increasing the number of antennas by another-of-magnitude (say, from 64 to 512), XL-arrays greatly expand the Rayleigh distance, which specifies the boundary between the near- and far-fields. This thus renders the communication users in future wireless systems are more likely to be located in the (radiative) \emph{near-field}, leading to a new research area, called near-field multiple-input-multiple-output (MIMO) communications~\cite{liu2023near}.

Compared with conventional far-field communications, near-field communications are featured by \emph{spherical} (instead of planar) wavefronts, where the channel steering vector is jointly determined by the user angle and range, which brings both opportunities and challenges. 
Interestingly, under the spherical wavefronts and the ideal assumption of continuous phase shifters (PSs), near-field beamforming based on maximum ratio transmission (MRT) exhibits the appealing \emph{beam-focusing} effect, which makes it possible to concentrate the beam power at specific locations/regions.
This is in contrast to the far-field case where the beamforming based on MRT can only steer the beam power along a certain angle. To study the beam-focusing performance, it was shown in~\cite{cui2022near} that, the beam-width of the near-field beam is inversely proportional to the array aperture (equivalently number of antennas). 
Moreover, a new metric, called \emph{beam-depth}, is further defined to characterize the distance region where the beam power is primarily concentrated, which is jointly determined by the user angle, range, wavelength and the antenna number~\cite{cui2022near,liu2023near}. 
Generally speaking, a larger number of antennas leads to a smaller beam-depth, thus enhancing the beam-focusing performance. 
In addition, for wideband communication systems, it was shown in~\cite{cui2024near} that, the MRT-based beamformer for one subcarrier can focus the beam power at a specific location only in this subcarrier, but it will be split into different locations for other subcarriers.
Besides the uniform linear array with half-wavelength inter-antenna spacing, near-field beamforming pattern has also been studied in other array configurations. For example, the authors in~\cite{zhou2024sparse} proposed to employ linear sparse arrays (LSAs) to enable near-field communication with relatively small number of antennas. 
They revealed that, in addition to the main lobe, LSAs introduce additional grating lobes, both of which exhibit the near-field beam-focusing effect.
Such beam pattern analysis is further extended to other non-uniform sparse arrays, such as modular arrays~\cite{li2024multi} and coprime arrays~\cite{zhou2024sparse,li2024sparse}.

On the other hand, the near-field beam pattern analysis motivates follow-up research to exploit the beam-focusing property to enhance the communication performance.
Specifically, apart from the angle domain control, the beam-focusing effect provides an additional degree-of-freedom (DoF) to flexibly control the beam power in the range domain \cite{wang2023extremely,liu2022deep}. This thus can be leveraged to enhance the received power at the target user, reduce inter-user interference (IUI), and improve multi-user access ability \cite{liu2023near,wu2023multiple}.
Moreover, other applications can also be enabled/enhanced by exploiting near-field beam-focusing effect. 
For example, wireless power transfer (WPT) efficiency can be greatly improved \cite{zhang2024swipt}, physical-layer security can be achieved even at the same angle \cite{zhang2024performance}, and target localization (including both angle and range estimation) can be achieved at a single XL-array without relying on multiple anchor nodes and inter-node synchronization \cite{pan2023ris,wang2024near}.

However, it is worth noting that the large number of antennas in XL-arrays also increase the hardware cost and power consumption, when fully digital beamforming architectures are considered.
This issue can be effectively addressed by employing hybrid beamforming structures, which include a high-dimensional analog beamformer and a low-dimensional digital beamformer. 
As for analog beamforming, continuous PSs are usually assumed, which can be approximately achieved in practice by using high-resolution (e.g., 4 bit) PSs. Nevertheless, this inevitably incurs high energy and hardware cost for XL-arrays.
For example, in millimeter-wave (mmWave) bands, the power consumption of a 4-bit PS is up to 45-108 mW in 60 GHz, which is much larger than those of low-resolution PSs, such as 5 mW for 1-bit PSs~\cite{mendez2016hybrid}. 
Therefore, one practical and more energy-efficient design for XL-arrays with half-wavelength inter-antenna spacing is employing \emph{low-resolution} PSs for analog beamforming.

Although uncharted in near-field communications, low-resolution PSs have been studied in the literature for far-field communications. 
For example, the authors in \cite{wang2018hybrid,sohrabi2016hybrid} proposed to optimize discrete phases for maximizing the far-field communication performance by using optimization theory and techniques. 
To reduce the computational complexity, an alternative method is by first designing the hybrid beamformer under the assumption of continuous PSs and then quantizing the obtained phase values into their discrete versions based on e.g., nearest neighborhood (NN) criterion \cite{liang2014low}. However, for near-field communications, especially the beam pattern analysis, low-resolution PSs introduce several new challenges. 
\begin{itemize}
\item First, it is unclear whether discrete PSs will affect the beam-focusing effect, which is conventionally revealed under the assumption of continuous PSs. 
\item Second, the phase quantization under discrete PSs renders the existing beam pattern analysis method inapplicable, since there generally lacks a closed-form expression for the analog beamformer under discrete PSs. 
\end{itemize}
These thus motivate the current work as the first attempt (to our best knowledge) to resolve the above issues for analytically characterizing the near-field 
beam pattern under discrete PSs.

\vspace{-7pt}
\subsection{Contributions, Organization, and Notations}
In this paper, we consider a multi-user near-field communication system as shown in Fig.~\ref{fig_1}, where an XL-array base station (BS) employs a hybrid beamforming architecture with discrete PSs. 
For ease of beam pattern analysis and obtaining useful insights, we consider 1) the two-stage hybrid beamforming method with its analog beamforming designed based on MRT, and 2) the NN phase quantization method under the constraints of discrete PSs, while the extensions to other scenarios will be discussed later.
The main contributions of this paper are summarized as follows.
\begin{itemize}
\item First, we propose a new and efficient Fourier series expansion (FSE) method to tackle the difficulty in characterizing
the beam pattern properties under phase quantization.
Specifically, by exploiting the fact that the near-field beam pattern under discrete PSs is a periodic function of the quantized phase, we re-express the beam pattern function into a more tractable summation form based on FSE, where only a few Fourier coefficients dominates and their values are fundamentally determined by the PS resolution.

\item Next, we analytically characterize the near-field beam pattern properties according to the dominant Fourier coefficients.
In particular, we show that discrete PSs introduce additional grating lobes, which are undesired in multi-user communications as they may cause IUI.
Second, it is revealed that the main lobe still exhibits the beam-focusing property.
Its beam power is determined by the Fourier coefficients which increases with PS resolution, while its beam-width and beam-depth are identical to those of continuous PSs.
In addition, there are two distinct types of grating lobes, individually characterized by the beam-focusing and beam-steering behaviors, respectively.
In addition, we provide intuitive understanding for the appearance of near-field grating lobes under discrete PSs from an \emph{array-of-subarrays} perspective, where the XL-array under discrete PSs is shown to have similar beam pattern with sparse arrays, thus resulting in grating lobes.

\item Finally, numerical results are presented to demonstrate the rate performance and energy efficiency (EE) of discrete PSs in the near-field. 
It is shown that, grating lobes generally degrade communication rate performance under MRT beamforming, which can be further mitigated by efficient digital beamforming. 
In addition, the rate gap between low-resolution PSs and continuous PSs gradually diminishes with the number of antennas, due to the increased beam-focusing gain and suppressed grating-lobe interference.
Moreover, low-resolution PSs exhibit much higher EE than that of high-resolution PSs, because they can significantly reduce the power consumption while at a certain rate loss.
In particular, it is shown that 3-bit PSs can achieve similar rate performance with the continuous PS counterpart, while achieving much higher EE.
\end{itemize}

\vspace{-0.5pt}
\textit{Organization:} The remainder of this paper is organized as follows. Section II presents the system model of XL-array under discrete PSs in the near-field. In Section III, we propose an FSE to facilitate the beam pattern analysis under $B$-bit PSs in the near-field.
The near-field beam pattern properties for discrete PSs are presented in Section IV, followed by discussions for special cases and intuitive explanation. 
Last, we present numerical results in Section VI and make concluding remarks in Section VII.

\vspace{-0.5pt}
\textit{Notations:} Lower-case and upper-case boldface letters represent vectors and matrices, respectively. Upper-case calligraphic letters (e.g., $\mathcal{N}$) denote discrete and finite sets. The superscripts $(\cdot)^T$ and $(\cdot)^H$ stand for the transpose and Hermitian transpose operations, respectively. $(x)^+$ denotes $\max \{0,x\}$. $\otimes$ is the Kronecker product. Moreover, $|\cdot|$ indicates the absolute value for a real number and the cardinality for a set. $\mathbb{Z}$ is the integer set and $\mathbb{C}^{M\times N}$ represents the space of $M\times N$ complex-valued matrices. $\mathcal{CN}(0,\sigma^2)$ is the distribution of a circularly symmetric complex Gaussian (CSCG) random vector with mean 0 and variance $\sigma^2$.

\section{System Model}
As shown in Fig.~\ref{fig_1}, we consider a downlink multi-user communication system in the narrow-band, where an XL-array BS equipped with an uniform linear array (ULA) serves $M$ single-antenna users. 
\vspace{-4pt}
\subsection{Channel Model}
Without loss of generality, the XL-array is placed at the $y$-axis and centered at the origin. 
Let $N = 2\tilde{N} -1 $ denote the number of antennas deployed at the XL-array. As such, the $n$-th antenna is located at $(0,nd)$, $n\in\mathcal{N}\triangleq \{ -\tilde{N}, -\tilde{N}+1,\cdots,0,\cdots,\tilde{N}\}$, where $d=\lambda/2$ is the inter-antenna spacing with $\lambda$ denoting the carrier wavelength.
In this paper, we consider the case where all the users are located in the Fresnel region, for which the distance from the XL-array center to user $m \in \mathcal{M}\triangleq\{1,2,\cdots,M\}$, denoted by $r_{m}$, satisfies $Z_{\rm F}=\max\{ d_{\rm R},1.2D\} \leq r_{m} \leq Z_{\rm R}=\frac{2D^2}{\lambda}$~\cite{zhou2024multi}. Herein, $Z_{\rm F}$ and $Z_{\rm R}$ denote the Fresnel distance and the Rayleigh distance, respectively, with $D$ denoting the XL-array aperture.
In particular, $d_{\rm R}$ represents the boundary between the reactive and radiative near-field regions, which is verified to be several wavelengths~\cite{ouyang2024impact}.
Thus, the Fresnel distance can be further obtained as $Z_{\rm F}=1.2D$.

Let $\mathbf{h}_{m}^H$ denote the channel from the XL-array to user ${m}$. In the Fresnel region with $r_{m}\geq Z_{\rm F}$, the channel amplitude variations across the antennas are generally negligible, while the phase variations are non-linear.
As such, under the general multi-path channel model,  $\mathbf{h}_{m}^H$ can be modeled as
\begin{equation}
\mathbf{h}_{m}^H\!=\!\sqrt{N}h_{m}\mathbf{b}^H(\theta_{m},r_{m})\!\!+\!\!\sqrt{\frac NL}\sum_{\ell=1}^{L_{m}}h_{m}^{(\ell)}\mathbf{b}^H(\theta_{m}^{(\ell)},r_{m}^{(\ell)}),\!\!
\end{equation}
which includes one line-of-sight (LoS) path and $L_{m}$ non-LoS (NLoS) paths. Herein, $\theta_{m} \in (-\frac{\pi}{2},\frac{\pi}{2})$ denotes the physical LoS angle-of-departure (AoD) from the BS to user $m$, and $h_{m} = \frac{\sqrt{\beta}}{r_m}e^{-\jmath\frac{2\pi}{\lambda} r_{m}}$ denotes the complex-valued channel
gain with $\beta$ representing the reference channel gain at a range of 1~m.  
The parameters $h_{m}^{(\ell)}$, $\theta_{m}^{(\ell)}$ and $r_{m}^{(\ell)}$ denote the complex channel gain, physical
angle and range of the $\ell$-th path, respectively. 
In this paper, we consider high-frequency bands such as mmWave and even terahertz (THz), for which the NLoS paths have negligible power due to  severe path-loss and shadowing \cite{8901159}. 
As such, the channel $\mathbf{h}_{m}^H$ can be approximated by its LoS component as
$
\mathbf{h}_{m}^H\approx \sqrt{N}h_{m}\mathbf{b}^H(\theta_{m},r_{m}),
$ \cite{cui2022channel}
where $\mathbf{b}^H(\theta_{m},r_{m})$ is the near-field steering vector under the spherical wavefront, which is given by 
\begin{equation}
\!\!\!\!\!\!\mathbf{b}\!\left(\theta_{m},r_{m}\right)\!\!=\!\!\frac{1}{\sqrt{N}}\!\bigg[e^{-\jmath \frac{2\pi}{\lambda} (r_{{m},-\tilde{N}}-r_{m})}\!,\!\cdots\!,e^{-\jmath \frac{2\pi}{\lambda} (r_{{m},\tilde{N}}-r_{m})}\!\bigg]^T\!\!\!\!.\!\!\!
\label{eq4}
\end{equation}
Based on the Fresnel approximation, the distance between user $m$ and the $n$-th XL-array antenna, denoted by $r_{{m},n}$, can be approximated~\cite{cui2022channel}
\begin{equation}
\begin{aligned}
r_{{m},n}&=\sqrt{r_{m}^2+n^2d^2-2r_{m}nd\sin\theta_{m} } \\
&\approx r_{m}-nd\sin\theta_{m}+\frac{n^2d^2\cos^2\theta_{m}}{2r_{m}},
\label{eq5}
\end{aligned}
\end{equation}

\subsection{Signal Model}
To reduce the hardware and energy cost, we consider the hybrid beamforming architecture for the XL-array BS~\cite{ahmed2018survey}. Let $\mathbf{F}_\mathrm{BB}\in\mathbb{C}^{N_{\mathrm{RF}}\times M}$ and $\mathbf{F}_\mathrm{RF}\in\mathbb{C}^{N\times N_{\mathrm{RF}}}$ denote the digital and analog beamformers, respectively. Then, the received signal at user $m$ is given by
$
y_{m}=\mathbf{h}_{m}^H\mathbf{F}_\mathrm{RF}\mathbf{F}_\mathrm{BB}\mathbf{x}+z_{m},
$ \cite{el2014spatially}
where $\mathbf{x}\in\mathbb{C}^{M\times1}$ denotes the transmitted data for $M$ users, and $z_{m}\sim\mathcal{CN}\left(0,\sigma_{m}^2\right)$ is the received additive white Gaussian noise (AWGN) with zero means and power $\sigma_{m}^{2}$.
To minimize the power consumption and without loss of generality, we assume that $N_\mathrm{RF}=M$.
\begin{figure}[!t]
\centering
	\includegraphics[width=0.85\linewidth]{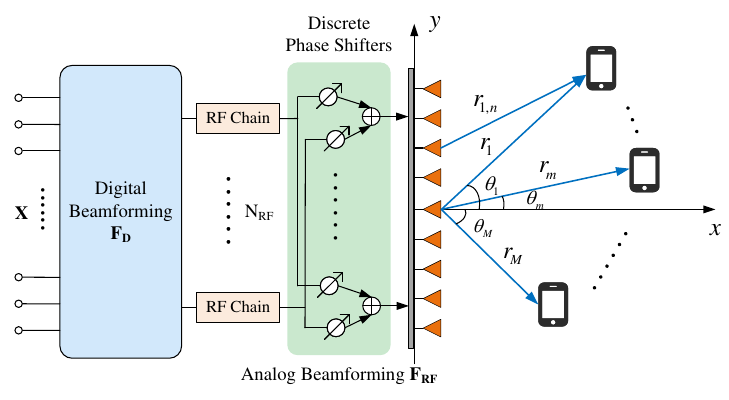}
	\centering
	\caption{An XL-array multi-user communication system under discrete PSs.}	
	\label{fig_1}
    \vspace{-12pt}
\end{figure}

Moreover, for ease of practical implementation, we consider \emph{discrete} PS for each antenna $n$, for which the phase of analog beamformer can only take a finite number of values equally sampled in $[0,2\pi)$. By denoting $B$ as the phase controlling bits, the set of available discrete phases at each antenna is obtained as $\mathcal{F}_{B}\triangleq\l\{\frac{(2{c}+1)\pi}{C_{B}}\mid c = 0,1,\cdots,C_{B}\!-\!1\r\}$, where $C_{B}=2^{B}$ denotes the number of quantization bins.\footnote{The obtained results can be further extended to the case with low-resolution digital-to-analog converters (DACs) in the conventional fully digital beamforming structure.}
Moreover, to study the effects of low-resolution PSs on the \emph{near-field beam-focusing} effect, we consider a low-complexity two-stage hybrid beamforming design~\cite{liu2023near}, where the analog beamforming is designed to maximize the received signal powers at individual users based on MRT by exploiting the near-field beam-focusing property. Then, digital beamforming is further devised to deal with residual IUI based on the effective channels accounting for analog beamforming, by using e.g., minimum mean square error (MMSE), zero-forcing (ZF) techniques. As such, we can obtain useful insights into the impacts of discrete PSs on the near-field beam pattern, which will be presented in the sequel.  

\section{Beam-Focusing Pattern Representation under Discrete Phase Shifters}

In this section, we propose an efficient FSE method to characterize the near-field beam pattern under the $B$-bit discrete PSs. Without loss of generality, we consider a typical user located at $(\theta_{\rm u}, r_{\rm u})$ and thus omit the user index in this section.
\subsection{Near-Field Beam Pattern under Phase Quantization}
Let ${\mathbf{f}_{\rm c}}(\theta_{\rm u},r_{\rm u})$ denote the MRT-based beamformer towards the location $(\theta_{\rm u},r_{\rm u})$ under the ideal assumption of \emph{continuous} phase shifts. As such, we have
\begin{align*}
{\mathbf{f}_{\rm c}}(\theta_{\rm u},r_{\rm u})&=\mathbf{b}(\theta_{\rm u},r_{\rm u})\nn\\
&=\frac{1}{\sqrt{N}}\left[e^{\jmath \varphi_{{\rm u},-\tilde{N}}},\cdots,e^{\jmath \varphi_{{\rm u},n}},\cdots,e^{\jmath \varphi_{{\rm u},\tilde{N}}}\right]^T,
\end{align*}
where $\varphi_{{\rm u},n} = \frac{2\pi}{\lambda}(nd\sin\theta_{\rm u}-\frac{n^2d^2\cos^2\theta_{\rm u}}{2r_{\rm u}})$. Let $U_{B}(\varphi)\in \mathcal{F}_{B}$ denote the phase-quantization function that maps a continuous phase to a discrete value according to a phase-quantization algorithm. 
Then, the MRT-based discrete beamformer under phase quantization is defined as
\begin{align*}
\mathbf{f}_{B}(\theta_{\rm u},r_{\rm u})\!\!=\!\!\frac{1}{\sqrt{N}}\!\!\left[e^{\jmath U_{B}(\varphi_{{\rm u},-\tilde{N}})},\!\cdots\!,e^{\jmath U_{B}(\varphi_{{\rm u},n})},\!\cdots\!,e^{\jmath U_{B}(\varphi_{{\rm u},\tilde{N}})}\right]^T\!\!.\!\!
\end{align*}
For ease of exposition, we define 
$
Q_{B}(\varphi)\triangleq e^{\jmath U_{B}(\varphi)},
$
where $Q_{B}(\cdot)$ is called the \emph{effective} quantization function and is simply referred to as the quantization function in the sequel by slight abuse of definition. Thus, we have
\begin{equation}
[\mathbf{f}_{B}(\theta_{\rm u},r_{\rm u})]_{n} = \frac{1}{\sqrt{N}}Q_{B}(\varphi_{{\rm u},n}),
\label{Eq:QB}
\end{equation}
whose near-field beam pattern is defined as below.

\begin{definition}[Beam pattern]\label{De:BPD}
{\rm
For the discrete analog beamformer $\mathbf{f}_{B}(\theta_{\rm u},r_{\rm u})$, its near-field beam pattern under $B$-bit PSs at an arbitrary
observation location $\{ r,\theta\}$ is defined as~\cite{zhou2024multi}
\begin{align}
F_{B}(\theta,r;\theta_{\rm u},r_{\rm u})&\triangleq\left|\mathbf{b}^H(\theta,r)\mathbf{f}_{B}(\theta_{\rm u},r_{\rm u}) \right|\nn\\
&=\frac1N\left|\sum_{n=-\tilde{N}}^{\tilde{N}}Q_{B}(\varphi_{{\rm u},n})e^{-\jmath \varphi_{ n}}\right|, \; \forall r, \theta.
\label{neq6}
\end{align}}
\end{definition}

Then, we provide the following definitions to characterize the beam pattern.~\cite{zhou2024sparse}. Note that different from continuous PSs, the near-field beam pattern under discrete PSs generally features both the main lobe and additional \emph{grating} lobes (which will be shown later). 
The definitions below are applicable to both the main and grating lobes.
\begin{definition}[Beam-width]
\rm The beam-width characterizes the angular width of the beam pattern in the \emph{distance-ring} (which is defined as $\!\!\{(r, \theta)|\frac{\cos^2\theta}{r}\!\!=\!\!k\frac{\cos^2\theta_{\rm u}}{r_{\rm u}}, k\!>\!0\}\!$ and shall be explained later in more details), where the beam power reduces to half of peak beam power (respectively denoted as $\theta_{\mathrm{right}}\!$ and $\theta_{\mathrm{left}}$). 
Mathematically, the beam-width is given by
\begin{equation}
\mathrm{BW}^{(B)}\!\triangleq\left|\sin\theta_{\mathrm{right}}-\sin\theta_{\mathrm{left}}\right|.
\end{equation}
\end{definition}

\begin{definition}[Beam-depth]
\rm At an observation angle $\theta$, the beam-depth characterizes the length of range interval $r \!\in\! [r_{\mathrm{large}}, r_{\mathrm{small}}]$, where the beam power reduces to half of peak beam power. Thus, the beam-depth is
\begin{equation}
\mathrm{BD}^{(B)}\triangleq\left|r_{\mathrm{large}}-r_{\mathrm{small}}\right|.
\end{equation}
\end{definition}

\begin{definition}[Beam-height]
\rm The beam-height characterizes the magnitude of the main lobe or grating lobe at an observation angle $\theta$, which is defined as 
\begin{equation}
\mathrm{BH}^{(B)}\triangleq\max_{r\in(Z_{\rm F}, Z_{\rm R})}\{F_{B}(\theta,r;\theta_{\rm u},r_{\rm u})\}.
\end{equation}
\end{definition}

Unlike the existing beam pattern analysis under continuous PSs \cite{cui2022channel}, the quantization function $Q_{B}(\cdot)$ in~\eqref{neq6} renders the beam pattern analysis under discrete PSs much more challenging, since there is no closed-form expression for $Q_{B}(\varphi_{{\rm u},n})$. 

\subsection{Beam Pattern Analysis Based on Fourier Series Expansion}
To address the above issue, we propose an efficient FSE method to characterize the near-field beam pattern for the discrete analog beamformer $\mathbf{f}(\theta_{\rm u},r_{\rm u})$.

Essentially, we exploit one key observation that $Q_{B}(\varphi)$ in~\eqref{Eq:QB} is a \emph{periodic} function of $\varphi$ with a period of $T=2\pi$. This motivates us to apply the FSE method to recast $Q_{B}(\varphi)$ in the following more tractable form
\begin{equation}
Q_{B}(\varphi)=\sum_{k=-\infty}^{\infty}a_k^{(B)}e^{\jmath k\frac{2\pi}{T}\varphi}=\sum_{k=-\infty}^{\infty}a_k^{(B)}e^{\jmath k\varphi},
\label{eq8}
\end{equation}
where the Fourier coefficient $a_k^{(B)}$ is given by
\begin{equation}
a_k^{(B)}\!\!=\!\dfrac{1}{T}\int_T \!\!Q_{B}(\varphi)e^{-\jmath k\varphi}d\varphi\!=\!\dfrac{1}{2\pi}\int_0^{2\pi}\!\!Q_{B}(\varphi)e^{-\jmath k\varphi}d\varphi.
\label{eq9}
\end{equation}
According to \eqref{eq8} and \textbf{Definition~\ref{De:BPD}}, we have the following result.

\begin{lemma}
\rm The near-field beam pattern of $\mathbf{f}(\theta_{\rm u},r_{\rm u})$ under $B$-bit PSs, given in \eqref{neq6}, can be equivalently expressed as
\begin{align}
F_{B}(\theta,r;\theta_{\rm u},r_{\rm u})\!&=\!\frac{1}{N}\!\left|\sum_{k=-\infty}^{\infty}\!\!a_{k}^{(B)}\sum_{n=-\tilde{N}}^{\tilde{N}}e^{\jmath (k\varphi_{{\rm u},n}-\varphi_{n})}\right|\nonumber\\ 
&=\!{\frac1N\!\left|\sum_{k=-\infty}^\infty\!\! a_k^{(B)}\!\!\sum_{n=-\tilde{N}}^{\tilde{N}}\!\!e^{\jmath \pi n\Delta_k+\jmath \frac{\pi\lambda}4n^2\Phi_k}\right|}\nonumber\\
&\triangleq\left|\sum\limits_{k=-\infty}^{\infty}f_k^{(B)}(\theta,r;\theta_{\rm u},r_{\rm u})\right|, \; \forall r, \theta,
\label{eq21}
\end{align}
where $\Delta_k\triangleq k\sin\theta_{\rm u}-\sin\theta$ is defined as the \emph{spatial angle
difference} and $\Phi_k\triangleq-k\frac{\cos^2\theta_{\rm u}}{r_{\rm u}}+\frac{\cos^2\theta}{r}$ is named as the \emph{ring
difference}~\cite{cui2022channel}. $f_k^{(B)}(\theta,r;\theta_{\rm u},r_{\rm u})$ in \eqref{eq21} is defined as
\begin{equation}
f_k^{(B)}(\theta,r;\theta_{\rm u},r_{\rm u})\triangleq\frac{1}{N}a_k^{(B)}\!\!\sum_{n=-\tilde{N}}^{\tilde{N}}\!\!e^{\jmath \pi n\Delta_k+\jmath \frac{\pi\lambda}4n^2\Phi_k}.
\label{neq14}
\end{equation}
\end{lemma}

\begin{remark} 
[Physical meaning of $\Phi_k$]\rm 
When $\Phi_k=C$ (a constant), we can get a set of locations which collectively form a ring where the locations $(\theta,r)$ satisfy $-k\frac{\cos^2\theta_{\rm u}}{r_{\rm u}}+\frac{\cos^2\theta}{r}=C$. 
For example, for the beam pattern under continuous PSs shown in Fig. \ref{figs:beam pattern}(a), the yellow dashed lines indicate the rings representing different values of $\Phi_1$. In particular, when $\Phi_1=0$ (i.e., $\frac{\cos^2\theta}{r}=\frac{\cos^2\theta_{\rm u}}{r_{\rm u}}$), the user location $(\theta_{\rm u}, r_{\rm u})$ is on the ring.
\end{remark}

To simplify the analysis and obtain valuable insights, we consider the NN method in the sequel, wherein each phase is quantized to its nearest neighbor based on the closest Euclidean distance criterion~\cite{liang2014low}. 

\begin{figure*}[t]
    \centering
    \captionsetup[subfigure]{justification=centering}
    \begin{subfigure}[b]{0.32\textwidth}
        \centering
        \includegraphics[width=\textwidth]{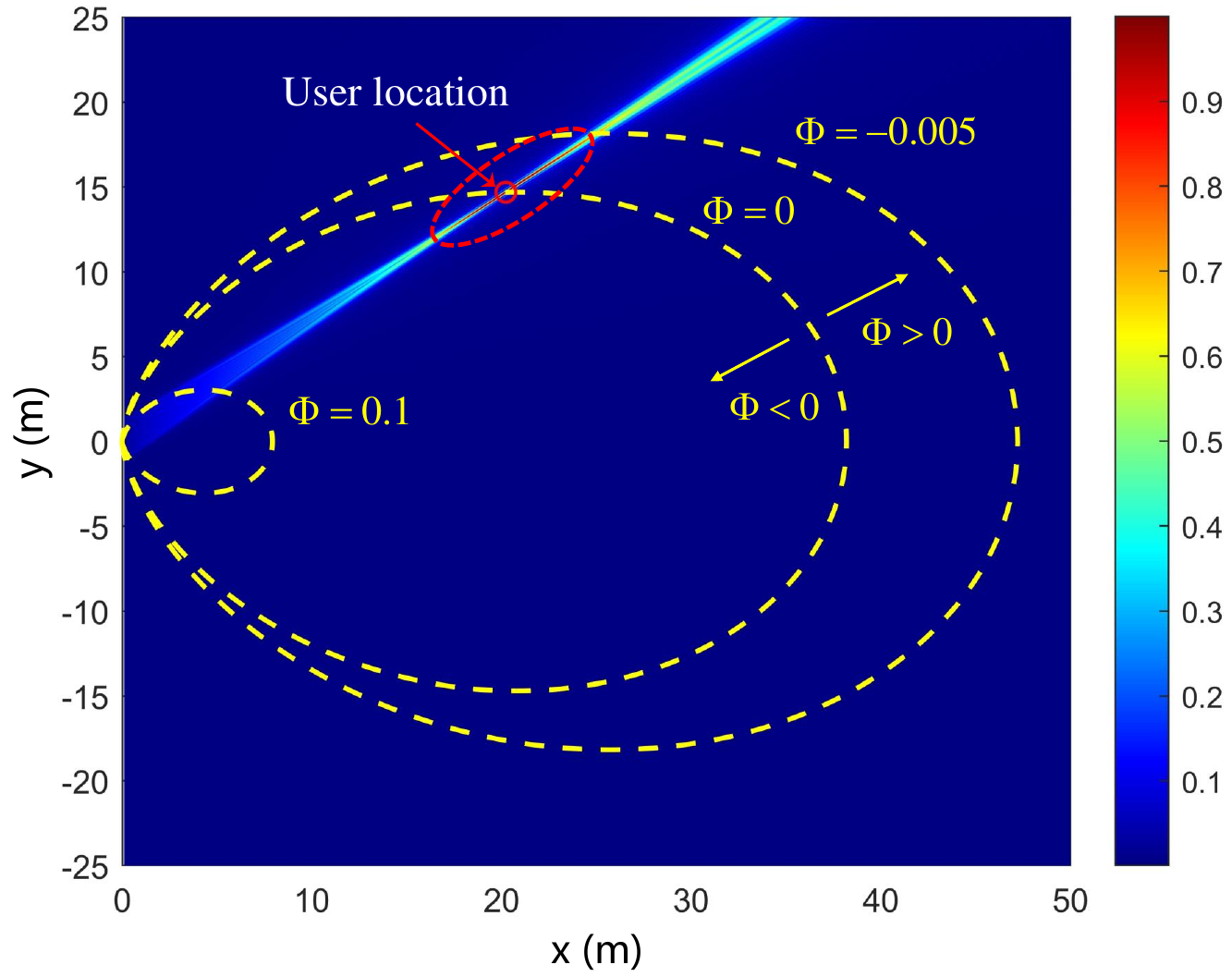} 
        \caption{Continuous PSs.}
    \label{fig.3.1}
    \end{subfigure}
    \hfill
    \begin{subfigure}[b]{0.32\textwidth}
        \centering
        \includegraphics[width=\textwidth]{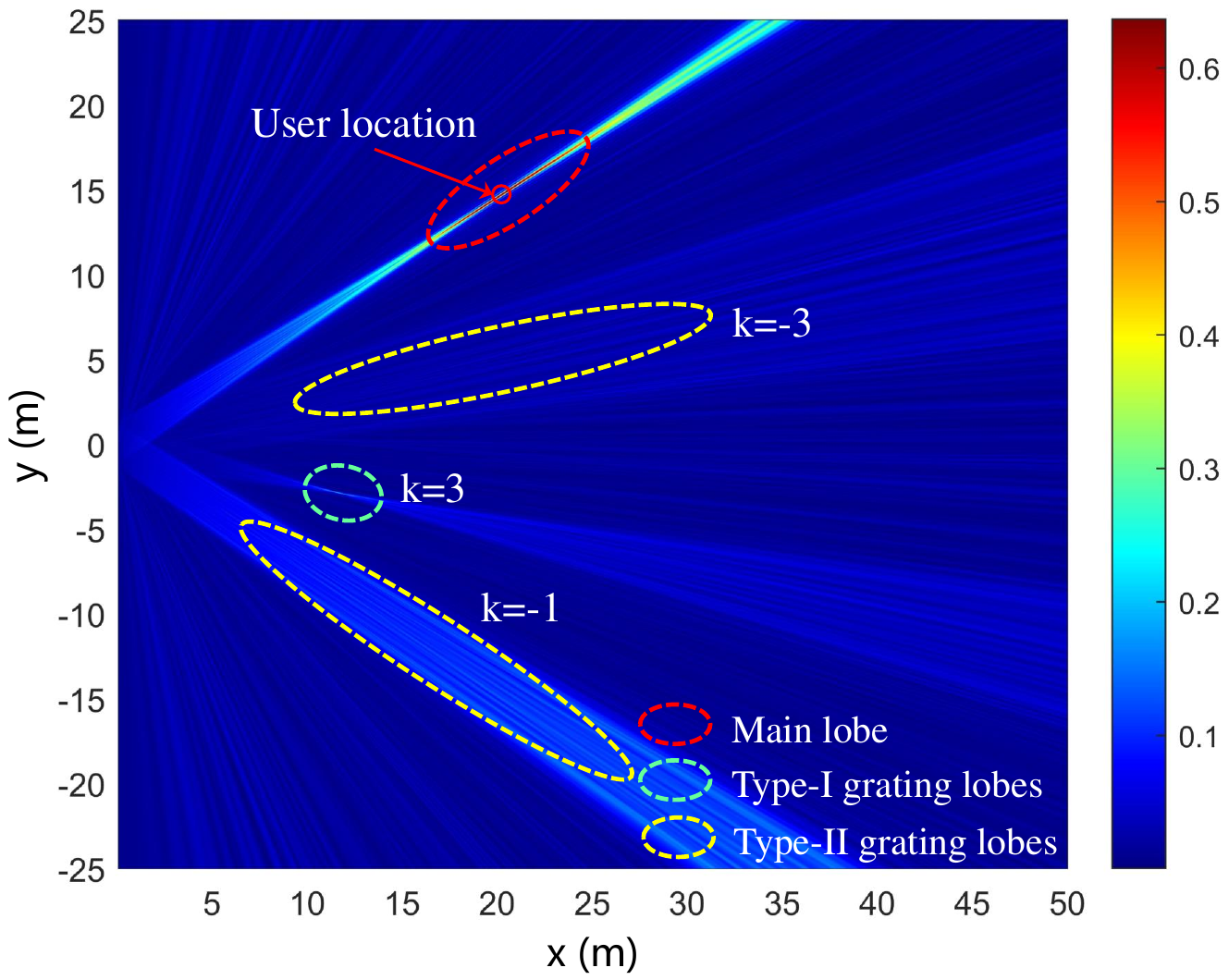} 
        \caption{1-bit PSs.}
        \label{fig.3}
    \end{subfigure}
    \hfill
    \begin{subfigure}[b]{0.32\textwidth}
        \centering
        \includegraphics[width=\textwidth]{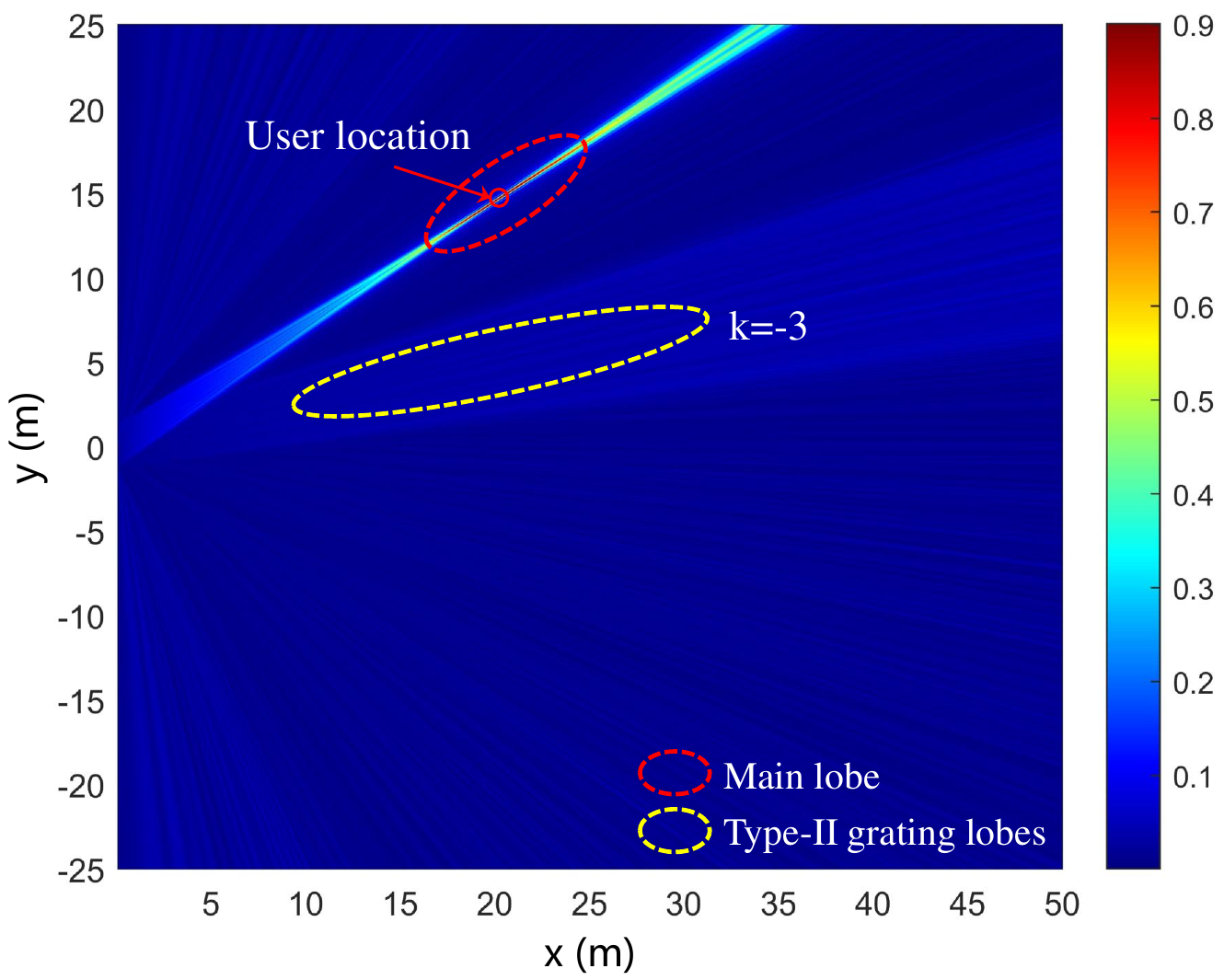}
        \caption{2-bit PSs.}
        \label{fig.6}
    \end{subfigure}
    \caption{The near-field beam pattern with $N=513$, $\theta_{\rm u}=\pi/5$, $r_{\rm u}=25$~m in different scenarios.}
    \label{figs:beam pattern}
    \vspace{-15pt}
\end{figure*}
\begin{lemma}[NN-based phase quantization]
\label{label:Quantization function}
\rm For a continuous phase $\varphi$, the discrete phase generated by the NN quantization method is $\hat{\varphi}=U_{B}(\varphi)=(2 \hat{c}+1)\pi/C_{B},$ where $\hat{c}$ is given by
\begin{equation}
\label{eq:Quantization function}
\hat{c} = \arg \underset{c \in \{0,1,\cdots,C_{B}\!-\!1\} }{\min} \left| \varphi_{\mathrm{mod}\;2\pi}-\frac{(2b+1)\pi }{C_{B}} \right|,
\end{equation}
with $\varphi_{{\rm mod}\;2\pi}=\mathrm{mod}(\varphi,2\pi)$ being the principal value of the continuous phase in the interval of $[0,2\pi)$.
\end{lemma}
\begin{remark}[Quantization methods]
\rm As shown in \cite{liang2014low}, the considered NN quantization method for analog beamforming can achieve superior rate performance when integrated with efficient digital beamforming methods (such as ZF), while it exhibits very low complexity in practice, hence striking flexible balance in performance-and-complexity trade-off. On the other hand, other more complicated phase quantization methods, such as the optimization techniques in~\cite{wang2018hybrid,sohrabi2016hybrid}, can be applied to further improve the rate performance at a higher computational complexity, for which the beam pattern analysis is more involved and thus left for our future work.
\end{remark}
Based on Lemma~\ref{label:Quantization function}, the Fourier coefficients of the NN method,  $\{a_k^{(B)}\}$, can be obtained in closed form as follows.
\begin{lemma}
\label{fourier coefficients}
\rm Given the $B$-bit PSs, the Fourier coefficients $\{a_k^{(B)}\}$ in~\eqref{eq9} under the NN quantization are given by
\begin{equation}
\begin{aligned}
\!\!a_k^{(B)}&\!\!=\!\!\left\{
  \begin{array}{ll}\!\!
    \frac{C_{B}}{k\pi}\sin(\frac{\pi}{C_{B}}), \!\! &\text{if } k\in \mathcal{T}, \\[12pt]\!\!
    0,  &\text{otherwise},
  \end{array}
\right.
\end{aligned}
\label{eq10}
\end{equation}
where $\mathcal{T}\triangleq \{t|t=1-pC_{B},\; p\in \mathbb{Z}\}$. 
\end{lemma}
\begin{proof}
\rm Please refer to Appendix A.
\end{proof}

Three important properties of $\!a_k^{(B)}\!$ are summarized as~follows. First, only a small number of $\!a_k^{(B)}$ have~non-zero values when $k$ is in the sparse set $\mathcal{T}$. Second, the maximum value of $|a_k^{(B)}|$ is given by $\max_k |a_k^{(B)}|$ $=|a_1^{(B)}|$. Third, $|a_k^{(B)}|$ decreases rapidly with $k$ in an order of $\mathcal{O}(1/k)$. 
Thus, for ease of analysis and maintaining high accuracy, we only consider \emph{dominant} terms which satisfy e.g., $|a_k^{(B)}|\!>\!0.1|a_1^{(B)}|$. This is equivalent to $|k|\!<\!10$ according to \eqref{eq10}.
As such, the beam pattern function in~\eqref{eq21} can be approximated as 
\begin{equation}
\label{nneq17}
F_{B}(\theta,r;\theta_{\rm u},r_{\rm u})\approx\left|\sum\limits_{k \in \mathcal{T'}}f_k^{(B)}(\theta,r;\theta_{\rm u},r_{\rm u})\right|,
\end{equation}
where $\mathcal{T'}= \mathcal{T} \cap [-9,9]$.
In other words, there is only a small number of $f_k^{(B)}(\theta,r;\theta_{\rm u},r_{\rm u})$ that need to be considered in \eqref{nneq17}.
Moreover, $|\mathcal{T'}|$ decreases with PS resolution, since $C_{B}$ increases with $B$, hence resulting in fewer terms that need to be considered in \eqref{nneq17}.
In particular, when the continuous PSs are employed, \eqref{eq8} reduces to $Q_{\infty}(\varphi)=a_1^{(\infty)}e^{\jmath \varphi}=e^{\jmath \varphi}$, which only contains one term $f_1^{(\infty)}(\theta,r;\theta_{\rm u},r_{\rm u})$ in~\eqref{nneq17}.

\section{Beam Pattern Characterization under Discrete Phase Shifters: General Case}
\label{General Case}
In this section, we characterize the near-field beam pattern under the general $B$-bit discrete PSs. To this end, we first present a set of important properties for each term $|f_k^{(B)}(\theta,r;\theta_{\rm u},r_{\rm u})|$ and then analyze the beam properties of the main lobe and grating lobes. Finally, we extend the results and present the far-field beam pattern under $B$-bit PSs.
\subsection{Properties of Principle Components} \label{Sec:Pinciple}
Note that it is difficult to directly analyze the beam pattern function \eqref{nneq17}, since it is a summation of several terms $f_k^{(B)}(\theta,r;\theta_{\rm u},r_{\rm u})$. To tackle this issue, we first present important properties of $f_k^{(B)}(\theta,r;\theta_{\rm u},r_{\rm u})$ to obtain useful insights. 

First, note that $|f_k^{(B)}(\theta,r;\theta_{\rm u},r_{\rm u})|$ can be expressed as
\begin{equation}
|f_k^{(B)}(\theta,r;\theta_{\rm u},r_{\rm u})|= |a_k^{(B)}||H(\Delta_k, \Phi_k)|,
\end{equation}
where $|H(\Delta_k, \Phi_k)|\triangleq\left|\frac{1}{N}\!\!\sum_{n=-\tilde{N}}^{\tilde{N}}\!e^{\jmath \pi n\Delta_k+\jmath \frac{\pi\lambda}4n^2\Phi_k}\right|$.
The coefficient $|a_k^{(B)}|$ dictates the scaling factor of the amplitude $|f_k^{(B)}(\theta,r;\theta_{\rm u},r_{\rm u})|$, which is affected by the PS resolution $B$. 
On the other hand, $|H(\Delta_k, \Phi_k)|$ characterizes the \emph{spatial behavior} of $|f_k^{(B)}(\theta,r;\theta_{\rm u},r_{\rm u})|$ which depends on $\Delta_k$ and $\Phi_k$ only, yet \emph{independent} of the PS resolution $B$. 

As such, we consider the case of continuous PSs to study the properties of $|H(\Delta_k, \Phi_k)|$ for convenience. Specifically, the beam pattern function under continuous PSs is given by
\begin{equation}
F_{\infty}(\theta,r;\theta_{\rm u},r_{\rm u})=|f_1^{(\infty)}(\theta,r;\theta_{\rm u},r_{\rm u})|= |H(\Delta_1, \Phi_1)|,
\end{equation}
where $\Delta_1\!\triangleq\!\sin\theta_{\rm u}\!-\!\sin\theta\!\in\!(\sin\theta_{\rm u}\!-\!1,\sin\theta_{\rm u}\!+\!1)$ 
and $\Phi_1\!\triangleq\!-\frac{\cos^2\theta_{\rm u}}{r_{\rm u}}\!+\!\frac{\cos^2\theta}{r}$.
Moreover, $\!|H(\Delta, \Phi)|\!$ is a periodic function of $\!\Delta\!$ with a period of 2. 
By numerically shown its beam pattern in Fig. \ref{figs:beam pattern}(a) with the beam-steering towards the location $(\theta_{\rm u}=\pi/5$, $r_{\rm u}=25~{\rm m})$, we obtain the following important observations.
First, $|H(\Delta, \Phi)|$ reaches its maximum value of 1 when $\Delta=0$ (i.e., $\theta\!=\!\theta_{\rm u}$) and $\Phi=0$.
Second, $|H(\Delta, \Phi)|$ monotonically decreases with~$|\Phi|$.
Third, $|H(\Delta, \Phi)|$ has large values around the angle of $\theta\!=\!\theta_{\rm u}$ (which satisfies $\Delta\!=\!0$); while $|H(\Delta, \Phi)|\!\approx\!0$ at other angles.
Based on the above, we have the following result.
\begin{property}
\label{principle term}
\rm The function $f_k^{(B)}(\theta,r;\theta_{\rm u},r_{\rm u})$ in~\eqref{neq14} has the following properties.\\
\textbf{(Periodicity)}
    $|f_k^{(B)}(\theta,r;\theta_{\rm u},r_{\rm u})|$ is a periodic function of $\Delta_k$ with a period of 2.\\
$\!$\textbf{(Sparsity)} $|f_k^{(B)}(\theta,r;\theta_{\rm u},r_{\rm u})|$ has large values at the angle
    \begin{equation}
    \theta_k =\arcsin(\mathrm{mod}(k\sin\theta_{\rm u}+1,2)-1).
    \label{eq23}
    \end{equation}
    In other angles, $|f_k^{(B)}(\theta,r;\theta_{\rm u},r_{\rm u})|$ has negligible value, i.e. $|f_k^{(B)}(\theta,r;\theta_{\rm u},r_{\rm u})|\approx 0$.\\
\textbf{(Spatial behavior)} According to the value of $k$, $\!|f_k^{(B)}\!(\theta,r;\theta_{\rm u},r_{\rm u}\!)|$ has different spatial behaviors. 
    \begin{enumerate}[label=\arabic*), leftmargin=0.7cm]
         \item For $k\!>\!0$, $|f_k^{(B)}(\theta,r;\theta_{\rm u},r_{\rm u})|$ achieves its maximum value $|a_k^{(B)}|$ at the following location
        \begin{align}
        \theta_k &=\arcsin(\mathrm{mod}(k\sin\theta_{\rm u}+1,2)-1),  \label{Eq:MVAngle}\\
        r_k & = \frac{\cos^2\theta_k}{k\cos^2\theta_{\rm u}}r_{\rm u}.\label{Eq:MVRange}
        \end{align}
        \item For $k\leq0$, $|f_k^{(B)}(\theta,r;\theta_{\rm u},r_{\rm u})|$ does not reach its maximum at any specific location. Instead, it steers beam pattern around $\theta_k$ given by \eqref{eq23}, which is similar to beam-steering in the far-field scenario.
    \end{enumerate}
\end{property}
\begin{proof}
Please refer to Appendix B.
\end{proof}

However, given these properties, it is still challenging to characterize the properties of $F_{B}(\theta,r;\theta_{\rm u},r_{\rm u})$ in \eqref{nneq17}, since it is \emph{jointly} affected by all components $\{f_k^{(B)}(\theta,r;\theta_{\rm u},r_{\rm u})\}$. To address this issue, in the following, we characterize the beam pattern properties of individual component $|f_k^{(B)}(\theta,r;\theta_{\rm u},r_{\rm u})|$ separately. Note that this corresponds to the case where the beam patterns of $\{f_k^{(B)}(\theta,r;\theta_{\rm u},r_{\rm u})\}$ for different $k$ do not overlap significantly, while the case with overlapping beam patterns is more difficult to characterize and will be discussed by numerical results in Section \ref{overlap}.

In particular, we define the (dominant) main lobe and grating lobes for $F_{B}(\theta,r;\theta_{\rm u},r_{\rm u})$ as follows.
\begin{itemize}
    \item {\bf Main lobe:} As $|f_1^{(B)}(\theta,r;\theta_{\rm u},r_{\rm u})|$ achieves its maximum value at the user location $(\theta_{\rm u},r_{\rm u})$ (see \eqref{Eq:MVAngle} and \eqref{Eq:MVRange}), we call it the main lobe of the beam pattern function.
    \item {\bf Grating lobes:} For $k\!\in\!\mathcal{G} \!\triangleq\!\{k\!\in\!\mathcal{T'}|k\!\neq1\}$, $|f_k^{(B)}(\theta,r;\theta_{\rm u},r_{\rm u})|$ achieves relatively large value around~$\theta_k$ given in \eqref{eq23}, which are thus regarded as the functions of grating lobes. 
    Based on the value of $k$, grating lobes can be classified into two categories: 1) $k>0$, which is referred to as Type-I grating lobes, which attain maximum values at specific locations (i.e., $(\theta_k, r_k)$ from \eqref{Eq:MVAngle} and \eqref{Eq:MVRange}), thereby exhibiting beam-focusing property; 2)  $k\leq0$, which is referred to as Type-II grating lobes that do not have a focused location and approximately exhibit the beam-steering property.
\end{itemize}

The phenomena of the main lobe and grating lobes for the near-field beam pattern under discrete PSs are observed in the Fig.~\ref{figs:beam pattern}(b) and (c) of 1-bit and 2-bit case, respectively.
The specific characteristics of these lobes are presented as follows.

\subsection{Properties of Main/Grating Lobes}
In this subsection, we analyze the characteristics of the main lobe and grating lobes for $F_{B}(\theta,r;\theta_{\rm u},r_{\rm u})$.
\subsubsection{Main lobe}
$\!$For the main-lobe function 
$|f_1^{(\!B)}\!(\theta,r;\theta_{\rm u},r_{\rm u}\!)|$, its beam characteristics are obtained as follows.
\begin{proposition}
\label{Main lobe}
\rm Consider the near-field beam pattern  under the $B$-bit discrete PSs, i.e., $F_{B}(\theta,r;\theta_{\rm u},r_{\rm u})$. The beam-height, beam-width and beam-depth of its main lobe are as follows.
\begin{itemize}
    \item The beam-height of the main lobe is 
    \begin{equation}
    \!\!\!\!\!\!\!\mathrm{BH}_1^{(B)}\! =\!\left|a_1^{(B)}\right|\!=\!\left|\frac{C_{B}}{\pi}\!\sin\l(\frac\pi {C_{B}}\r)\right|\!=\!\left|\frac{2^{B}}\pi\!\sin\l(\frac\pi{2^{B}}\r)\right|.\!\!
    \label{neq23}
    \end{equation}
    \item The beam-width of the  main lobe is 
    $\mathrm{BW}_1^{(B)} = \frac{1.76}{N}.$
    \item The beam-depth of the  main lobe is 
    \begin{equation}
    \label{BD_main_lobe}\mathrm{BD}_1^{(B)}=\left\{\begin{array}{ll}
    \frac{2r_{\rm u}^2r_\mathrm{DF}}{r_\mathrm{DF}^2-r_{\rm u}^2}, & r_{\rm u}<r_{\mathrm{DF}}, \\
    \infty, & r_{\rm u} \geq r_{\mathrm{DF}},
    \end{array}\right.
    \end{equation}
    where $r_\mathrm{DF} = \frac{N^2\lambda \cos^2\theta_{\rm u}}{8\eta^2_{\mathrm{3dB}}}$ and $\eta_{\mathrm{3dB}}=1.31$.
\end{itemize}
\end{proposition}
\begin{proof}
First, the beam-height of the main lobe is its maximum value, i.e., $\max_{\{\theta,r\}} |f_1^{(B)}(\theta,r;\theta_{\rm u},r_{\rm u})|=|a_1^{(B)}|$, which is obtained from~\eqref{eq10} with $k=1$. Then, for the proof of beam-width and beam-depth, please refer to Appendices C and D for the case of $k=1$, respectively.
\end{proof}

Note that, according to \eqref{BD_main_lobe}, when $r_{\rm u}<r_{\mathrm{DF}}$, the main lobe exhibits a beam-focusing property. As such, the region where $Z_{\rm F}\leq r_{\rm u}<r_{\mathrm{DF}}$ is referred to as the \emph{effective near-field region}.

\begin{figure}[!t]
\vspace{-10pt}
\centering
\includegraphics[width=0.75\columnwidth]{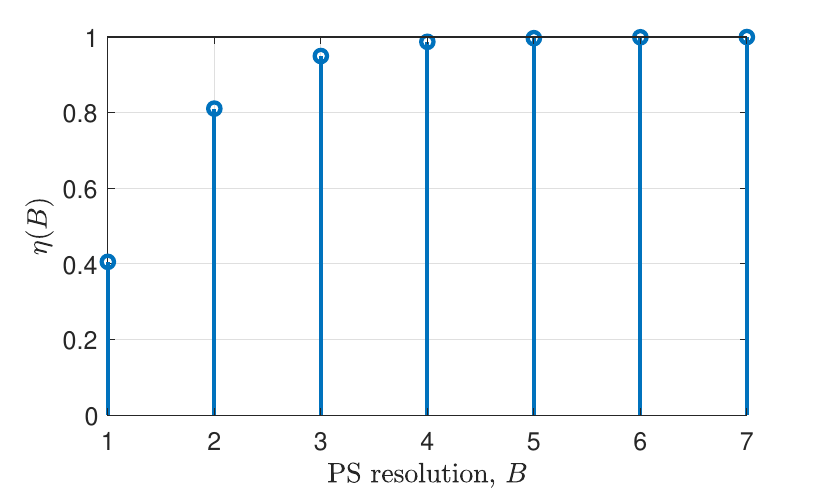}
\caption{The value of $\eta({B})$ versus the PS resolution $B$.}
\label{fig.7}
\vspace{-15pt}
\end{figure}
\begin{remark}[Impact of discrete PSs on main lobe] 
\rm \textbf{Proposition~\ref{Main lobe}} shows that both the beam-width and beam-depth of the main lobe under discrete PSs are \emph{not} affected by the discrete PS resolution, and hence are the \emph{same} as those under continuous PSs. This can be intuitively understood, since the beam pattern expressions of the main lobe (i.e., $|f_1^{(B)}(\theta,r;\theta_{\rm u},r_{\rm u})|=|a_1^{(B)}||H(\Delta_1, \Phi_1)|$) and that of continuous one (i.e., $F_{\infty}(\theta,r;\theta_{\rm u},r_{\rm u})=|H(\Delta_1, \Phi_1)|$) differ only in their coefficients $|a_1^{(B)}|$. Therefore, discrete PS resolution affects the beam-height of the main-lobe (see $|a_1^{(B)}|$ in \eqref{neq23}), while we have $|a_1^{(\infty)}|\!=\!1$ for continuous PSs. Let $\eta({B})$ denote the power ratio between the beam-heights under the discrete and continuous PSs, which is given by
\begin{equation}
\eta({B})=\frac{|a_1^{(B)}|^2}{|a_1^{(\infty)}|^2}
=\left(\frac{2^{B}}\pi\sin\left(\frac\pi{2^{B}}\right)\right)^2.
\end{equation} 
As shown in Fig.~\ref{fig.7}, the beamformer under 1-bit PSs suffers substantial power (or beam-height) loss as compared to the continuous case, while the loss reduces as the PS resolution increases and becomes negligible when $B\geq3$.
\end{remark}

\subsubsection{Grating lobes}
According to Section~\ref{Sec:Pinciple}, $\{|f_k^{(B)}(\theta,r;\theta_{\rm u},r_{\rm u})|, k \in \mathcal{G}\}$ represent grating lobes.
In the next, the beam properties of Type-I and Type-II grating lobes are characterized as follows.

\underline{\textbf{Type-I grating lobes:}} For $k\in \mathcal{G}_1\triangleq \{k\in \mathcal{G}| k>1\}$, we show below that Type-I grating lobes concentrate the beam power at specific locations.
\begin{proposition}
\label{Type-I grating lobes}
\rm Under the $B$-bit PSs, the beam-widths, beam-depths, and beam-heights of Type-I grating lobes, i.e. $|f_k^{(B)}(\theta,r;\theta_{\rm u},r_{\rm u})|$ with $k\in \mathcal{G}_1$, are given as follows.
\begin{itemize}
    \item The beam-heights of Type-I grating lobes are
    \begin{align}    \mathrm{BH}_k^{(B)}&= |a_k^{(B)}|=\left|\frac{C_{B}}{k\pi}\sin\left(\frac{\pi}{C_{B}}\right)\right| \\   &=\left|\frac{2^{B}\sin\left(\frac{\pi}{2^{B}}\right)}{(1-p2^{B})\pi}\right|, k\in \mathcal{G}_1, p\in \mathcal{P},
    \label{neq34}
    \end{align}
    where $\mathcal{P}\triangleq \{p|1-p2^{B}\in \mathcal{G}_1\}$.
    \item The beam-widths of Type-I grating lobes are
    \begin{equation}
    \mathrm{BW}_k^{(B)}=\frac{1.76}{N}, \; k\in \mathcal{G}_1.
    \end{equation}
    \item The beam-depths of Type-I grating lobes are
    \begin{equation}
    \!\!\mathrm{BD}_k^{(B)}\!\!=\!\!\left\{\begin{array}{ll}\!
    \frac{\cos^2\theta_k}{\cos^2\theta_{\rm u}}\frac{2r_{\rm u}^2r_\mathrm{DF}}{k^2r_\mathrm{DF}^2-r_{\rm u}^2}, &\!\! r_{\rm u}<kr_{\mathrm{DF}}, \\
    \infty, &\!\! r_{\rm u} \geq kr_{\mathrm{DF}},
    \end{array}\right. k\in \mathcal{G}_1.
    \end{equation}
    
\end{itemize}
\end{proposition}
\begin{proof}
\rm First, the beam-height of the Type-I grating lobe is its maximum value, i.e., $|a_k^{(B)}|$ for $k\in\mathcal{G}_1$, which can be obtained from~\eqref{eq10}. 
Then, for the proof of beam-width and beam-depth, please refer to Appendices C and D for the case of $k>1$, respectively.
\end{proof}

Note that for Type-I grating lobes, their beam-depths are affected by the focused location of the main lobe (i.e., $(\theta_{\rm u}, r_{\rm u})$) and the beam-heights are decided by the PS resolution, while their beam-widths are determined by the antenna number only.
Moreover, given $p$, the beam-height in \eqref{neq34} decreases with $B$, since its denominator increases exponentially with $B$ and the numerator increases slowly with $B$ as shown in Fig.~\ref{fig.7}.
\begin{corollary}[Asymptotic characteristic of main lobe and Type-I grating lobes]
\rm When the number of antennas $N$ becomes sufficiently large, the main lobe and Type-I grating lobes in the near-field focus their power exclusively on the point $(\theta_k, r_k)$, as given by \eqref{Eq:MVAngle} and \eqref{Eq:MVRange}, with no energy distributed in other regions, i.e.,
\begin{align}
\lim_{N\to\infty}|f_k^{(B)}(\theta,&r;\theta_{\rm u},r_{\rm u})|=0, \notag\\
&\text{ for } r\neq r_k \text{ or } \theta\neq\theta_k, k \in (\mathcal{G}_1 \cup \{1\}).
\label{corollary_1}
\end{align}
\end{corollary}
\begin{proof}
By using a method similar to that in~\cite{wu2023multiple}, we can obtain the desired results in \eqref{corollary_1}. First, we replace the summation in $|f_k^{(B)}(\theta,r;\theta_{\rm u},r_{\rm u})|$ with an integral. Then, we can show that when $r\neq r_k$ or $\theta\neq\theta_k$, the integral converges to 0, regardless of whether $\Phi_k=0$ or not, which completes the proof.
\end{proof}

\underline{\textbf{Type-II grating lobes:}} For $k\in \mathcal{G}_2\triangleq \{k\in \mathcal{G}| k\leq0\}$, we show below that Type-II grating lobes approximately steer the beam power
at certain angles, instead of exhibiting the beam-focusing behavior.
Note that for this case, there is no need to characterize the beam-heights, beam-widths and beam-depths, due to the disappearance of beam-focusing property.
As an alternative, we use surrogate beam-width defined below to characterize the angular width, as the peak beam power is difficult to obtain.
\begin{definition}[Surrogate beam-width]
\label{surrogate}
\rm This surrogate beam-width for the Type-II grating lobe at angle $\theta_k$ and distance $r_0$ is defined as the angular width at the ring ($\{(\theta,r)|\frac{\cos^2\theta}{r}=\frac{\cos^2\theta_k}{r_0}\}$), where the beam power reduces to half of the beam power at angle $\theta_k$~\cite{wu2024near}, which is respectively denoted as ${\Omega}_{\mathrm{left}}$ and ${\Omega}_{\mathrm{right}}$. As such, the surrogate beam-width is given by
\begin{equation}
\mathcal{A}^{(B)}_k(r_0)\triangleq\left|\sin\Omega_{\mathrm{right}}-\sin\Omega_{\mathrm{left}}\right|.
\end{equation}
\end{definition}
The surrogate beam-width defined above uses the power at the angle $\theta_k$ as the reference peak power, thus making the beam-width analysis for Type-II grating lobes more tractable. 
\begin{proposition}
\label{Type-II grating lobes}
\rm Under $B$-bit PSs, the properties of beam power and surrogate beam-width of Type-II grating lobe, i.e. $|f_k^{(B)}(\theta,r;\theta_{\rm u},r_{\rm u})|$ with $k\in \mathcal{G}_2$, are given as follows.
\begin{itemize}
    \item The beam power of Type-II grating lobe, i.e., $|f_k^{(B)}(\theta,r;\theta_{\rm u},r_{\rm u})|^2$, increases with $r$ at the angle $\theta_k$ given in \eqref{eq23}, while it decreases with $N$ and $B$. Moreover, given $r_{\rm u}, N$ and $B$, the beam power is upper-bounded as $|f_k^{(B)}(\theta_k,r;\theta_{\rm u},r_{\rm u})|^2<| a_k^{(B)}|^2 \left| \frac{C(\beta_{\infty})+\jmath S(\beta_{\infty})}{\beta_{\infty}} \right|^2$.
    \item When the user is located in the effective near-field region\footnote{When user range is beyond the effective near-field region, the surrogate beam-width is much narrower, and thus the closed-form expression in \eqref{surrogate beam-width} may not be accurate. However, it can still be obtained through numerical methods.} (i.e. $Z_{\rm F}\leq r_{\rm u}<r_{\mathrm{DF}}$), the surrogate beam-width of Type-II grating lobe at distance $r_0$ is given by
    \begin{equation}
    \!\mathcal{A}^{(B)}_k(r_0)\!=\!Nd|\Phi_{k,0}|+(1-\sqrt{3})\sqrt{d|\Phi_{k,0}|}, \; k\in \mathcal{G}_2,
    \label{surrogate beam-width}
    \end{equation}
    where $\Phi_{k,0}\!=\!-k\frac{\cos^2\theta_{\rm u}}{r_{\rm u}}+\frac{\cos^2\theta_k}{r_0}$. Moreover, the surrogate beam-width decreases with $r_0$ and increases with $N$.
\end{itemize}
\end{proposition}
\begin{proof}
\rm Please refer to Appendix E.
\end{proof}
\begin{corollary}[Asymptotic characteristic of Type-II grating lobes]
\vspace{-4pt}
\rm When the number of antennas $N$ becomes sufficiently large, the power of Type-II grating lobe in the near-field approaches zero, i.e.,
\begin{align}
\lim_{N\to\infty}|f_k^{(B)}(\theta,r;\theta_{\rm u},r_{\rm u})|=0,\; \forall r, \theta, k\in\mathcal{G}_2. 
\end{align}
\end{corollary}
\begin{proof}
The proof is similar to Corollary 1 and thus is omitted.
\end{proof}
\begin{remark}[Effect of PS resolution]
\vspace{-10pt}
\rm Comparing \textbf{\textbf{Propositions}} 2 and 3, we can conclude that the PS resolution (i.e., $B$) affects the beam power of both the main and grating lobes, while it has no influence on the beam-width and beam-depth. In general, the higher the PS resolution, the lower the power loss in the main lobe due to discrete PSs. In addition, the beam powers of grating lobes are suppressed more effectively when the PS resolution increases.
\end{remark}
\vspace{-10pt}
\subsection{Far-field Beam Pattern under Discrete PSs}
\label{far-field}
In this subsection, we extend the above beam pattern analysis method to the far-field case to obtain useful insights. Note that when $r_{\rm u}$ and $r$
are sufficiently large, the near-field channel model degenerates into the far-field counterpart. 
Specifically, the far-field MRT-based discrete beamformer under $B$-bit phase quantization is given by
$
[\mathbf{f}_{\mathrm{far},B}(\theta_{\rm u})]_{n} = \frac{1}{\sqrt{N}}Q_{B}(n\pi\sin\theta_{\rm u}).
$
Then, the main/grating lobe function can be obtained as
\begin{align}
\!\!|f_{\mathrm{far},k}^{(B)}(\theta;\theta_{\rm u})|&\!=\!\left|\frac{1}{N}a_k^{(B)}\!\!\sum_{n=-\tilde{N}}^{\tilde{N}}\!\!e^{\jmath \pi n\Delta_k}\right| \nonumber\\ 
&\!=\!\left|a_{k}^{(B)}\Xi_{N}(\pi(k\sin\theta_{\rm u}-\sin\theta))\right|\!,\;k \in \mathcal{T'}.
\label{far field beam pattern}
\end{align}
As such, the beam pattern properties for the far-field case under the $B$-bit PSs can be obtained by using the similar approach to the near-field case, which are presented below.
\begin{proposition}
\label{far-field main lobe}
\rm The main lobe of the far-field discrete beamformer $\mathbf{f}_{\mathrm{far},B}(\theta_{\rm u})$ steers the beam power along the user angle $\theta_{\rm u}$. Moreover, its beam-height is $\mathrm{BH}_{\mathrm{far},1}^{(B)} = |a_1^{(B)}|$ and its beam-width is $\mathrm{BW}_{\mathrm{far},1}^{(B)} = \frac{1.76}{N}$. 
\end{proposition} 
\begin{proof}
The main lobe beam function in the far-field is given by $|f_{\mathrm{far},1}^{(B)}(\theta;\theta_{\rm u})|=\left|a_{1}^{(B)}\Xi_{N}(\pi(\sin\theta_{\rm u}-\sin\theta))\right|$.
Then, the beam-height of the main lobe is its maximum value, i.e., $\max_\theta|f_{\mathrm{far},1}^{(B)}(\theta,\theta_{\rm u})|\!=\!|a_1^{(B)}|$. Moreover, we can use the same method as in the proof of beam-width in \textbf{Proposition 1} to obtain the result $\mathrm{BW}_{\mathrm{far},1}^{(B)} = \frac{1.76}{N}$.
\end{proof}
\begin{proposition}
\label{far-field grating lobe}
\rm The grating lobes of the far-field discrete beamformer $\mathbf{f}_{\mathrm{far},B}(\theta_{\rm u})$ exhibit the beam-steering property and occur at the angles 
$
\theta_k =\arcsin(\mathrm{mod}(k\sin\theta_{\rm u}+1,2)-1),\; k\in \mathcal{G}.
$
Their beam-heights are given by $\mathrm{BH}_{\mathrm{far},k}^{(B)} =|a_k^{(B)}|$ and their beam-widths are $\mathrm{BW}_{\mathrm{far},k}^{(B)} = \frac{1.76}{N},\; \forall k\in \mathcal{G}$.
\end{proposition}
\begin{proof}
The grating lobe functions in the far-field are given by $|f_{\mathrm{far},k}^{(B)}(\theta;\theta_{\rm u})|=\left|a_{k}^{(B)}\Xi_{N}(\pi(\sin\theta_{\rm u}-\sin\theta))\right|$, $\forall k\in\mathcal{G}$.
They can achieve their maximum
values $\{|a_k^{(B)}|\}$ at angles $\theta_k =\arcsin(\mathrm{mod}(k\sin\theta_{\rm u}+1,2)-1),\; \forall k\in \mathcal{G}$. As such, these grating lobes steer the beam power at these angles $\{\theta_k\}$ and beam-heights $|a_k^{(B)}|$ are obtained.
Moreover, we can use the same method as in the proof of beam-width in \textbf{Proposition 2} to obtain similar results of $\mathrm{BW}_{\mathrm{far},k}^{(B)} = \frac{1.76}{N}$, $\forall k\in \mathcal{G}$.
\end{proof}
    

Note that, different from the near-field scenario, there is only one type of grating lobe, and all grating lobes exhibit the beam-steering property in the far-field.
\begin{remark}[Far-field versus near-field]
\rm From the beam pattern perspective, employing discrete PSs in the near-field demonstrates superior performance over the far-field scenario, 
due to two main factors.
\begin{itemize}
    \item In the far-field, the main lobe and Type-I grating lobes exhibit beam-steering rather than the beam-focusing observed in the near-field, which significantly increases IUI at the same angle.
    \item In the far-field, the grating lobes associated with $k<0$ exhibit higher beam power than Type-II grating lobes in the near-field, hence potentially resulting in stronger IUI.
\end{itemize}
\end{remark}
\begin{figure}[!t]
  \centering
    \includegraphics[width=0.7\columnwidth]{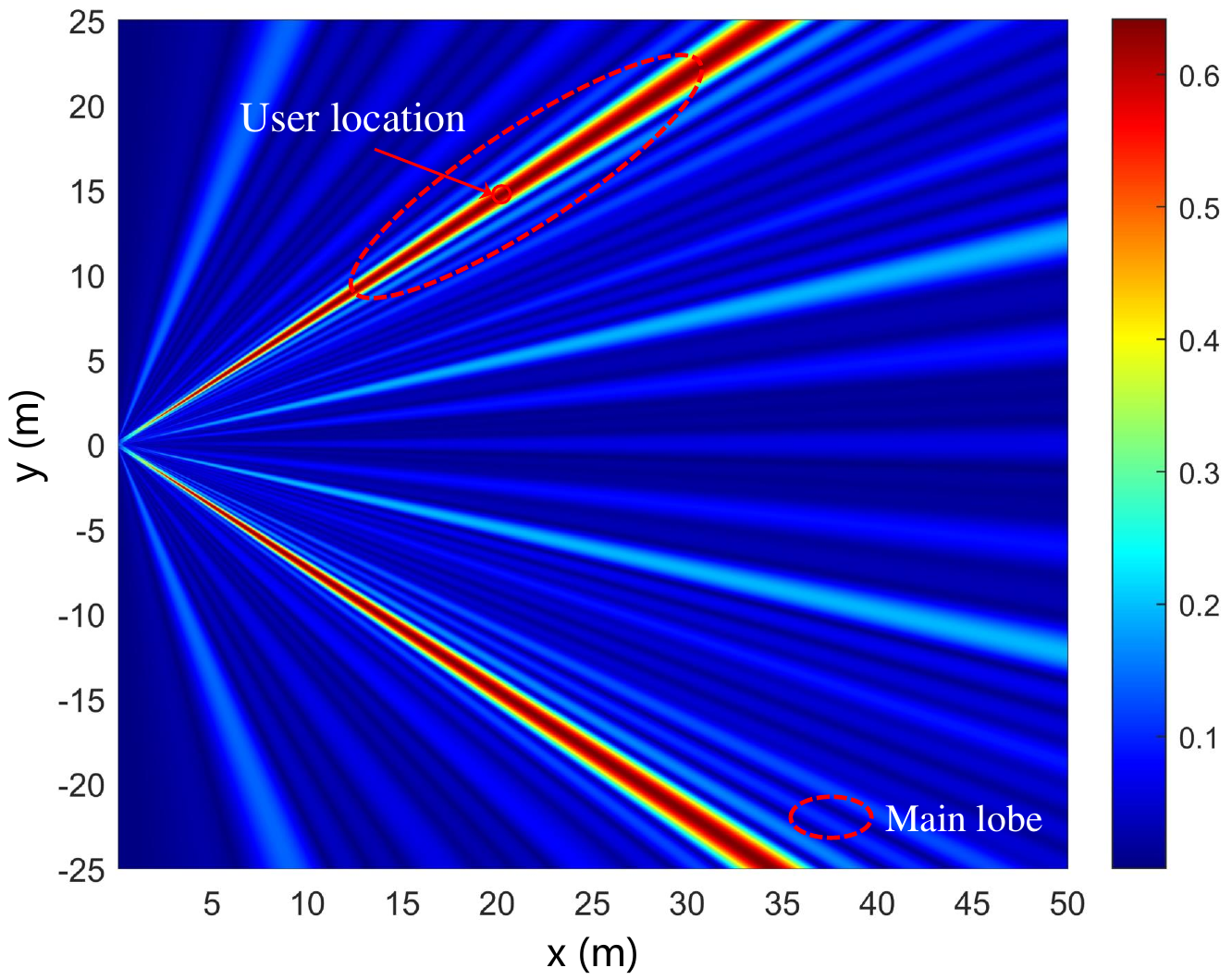}
    \caption{The far-field beam pattern with $N=65$, $\theta_{\rm u}=\pi/5$, $r_{\rm u}=25$ m under 1-bit PSs.}
    \label{far-field beam pattern}
    \vspace{-15pt}
\end{figure}
\section{Beam Pattern Characterization under Discrete Phase Shifters: Special Cases and Discussions}
In this section, we first present useful insights for the special cases of 1-bit and 2-bit PSs. Then, we provide physical interpretation for the far- and near-field beam patterns under discrete PSs from an \emph{array-of-subarray} perspective.
Finally, we analyze the impact of discrete PSs on multi-user communication performance.

\subsection{Beam Pattern under 1-bit and 2-bit PSs}
In Fig.~\ref{figs:beam pattern}(b), we present the near-field beam pattern for 1-bit PSs under the setup of $\{N\!=\!513, \;\!\theta_{\rm u}\!=\!\pi/5,\;\! r_{\rm u}\!=\!25\;\!\mathrm{m}\}$.
Several interesting observations are made as follows. 
First, for the 1-bit PSs, there exist multiple grating lobes which have  different beam powers and beam shape. 
Specifically, Type-I grating lobes (e.g., $k\!=\!3$) exhibit the beam-focusing property and their beam-heights are $|a_{k}^{(1)}|$, while Type-II grating lobes (e.g., $k\!=\!-1$ and $k\!=\!-3$) do not have the beam-focusing property.
In particular, the grating lobe with $k = -1$ has the highest power $|a_{-1}^{(1)}| = |a_1^{(1)}|$, which occurs at $\theta = -\theta_{\rm u}$ (see~\eqref{Eq:MVAngle}). 
This is referred to as the \emph{primary grating lobe} for the 1-bit case.
Moreover, the main lobe is centered around the user location, exhibiting a beam-focusing property with a beam-height of $|a_1^{(1)}|\approx\frac{2}{\pi}$.

Next, we further show the near-field beam pattern under 2-bit PSs in Fig.~\ref{figs:beam pattern}(c). It is observed that the grating lobes are much weaker than the main lobe, as $|a_{1}^{(2)}|$ is much higher than the others with $k\neq1$.
Notably, its grating lobes do not have strong power as the primary grating lobe in the 1-bit case.
Moreover, the beam-height of the main lobe for the 2-bit PS case is much larger than that of the 1-bit case.

In addition, for the near- and far-field beam pattern comparison, we present the far-field beam pattern under 1-bit PSs in Fig.~\ref{far-field beam pattern} given $\theta_{\rm u}\!=\!\pi/5$, $r_{\rm u}\!=\!25$ m and $N\!=65$.
One can observe that for the far-field case, all the main lobe and grating lobes exhibit the beam-steering property, sharing the same beam-heights of $|a_{k}^{(1)}|$.
Moreover, the grating lobes with $k<0$ under the far-field case exhibit higher beam powers than Type-II grating lobes in the near-field.
\subsection{Overlap among Dominant Lobes}
\label{overlap}
Next, we discuss the impact of beam overlap among dominant lobes on the near-field beam pattern.
When the angles of these lobes overlap significantly, the beam characteristics are the superposition of individual beam patterns. 
However, it is generally challenging to obtain a closed-form expression for the composite beam pattern (if tractable).

\begin{figure}[!t]
    \centering
    \includegraphics[width=0.7\columnwidth]{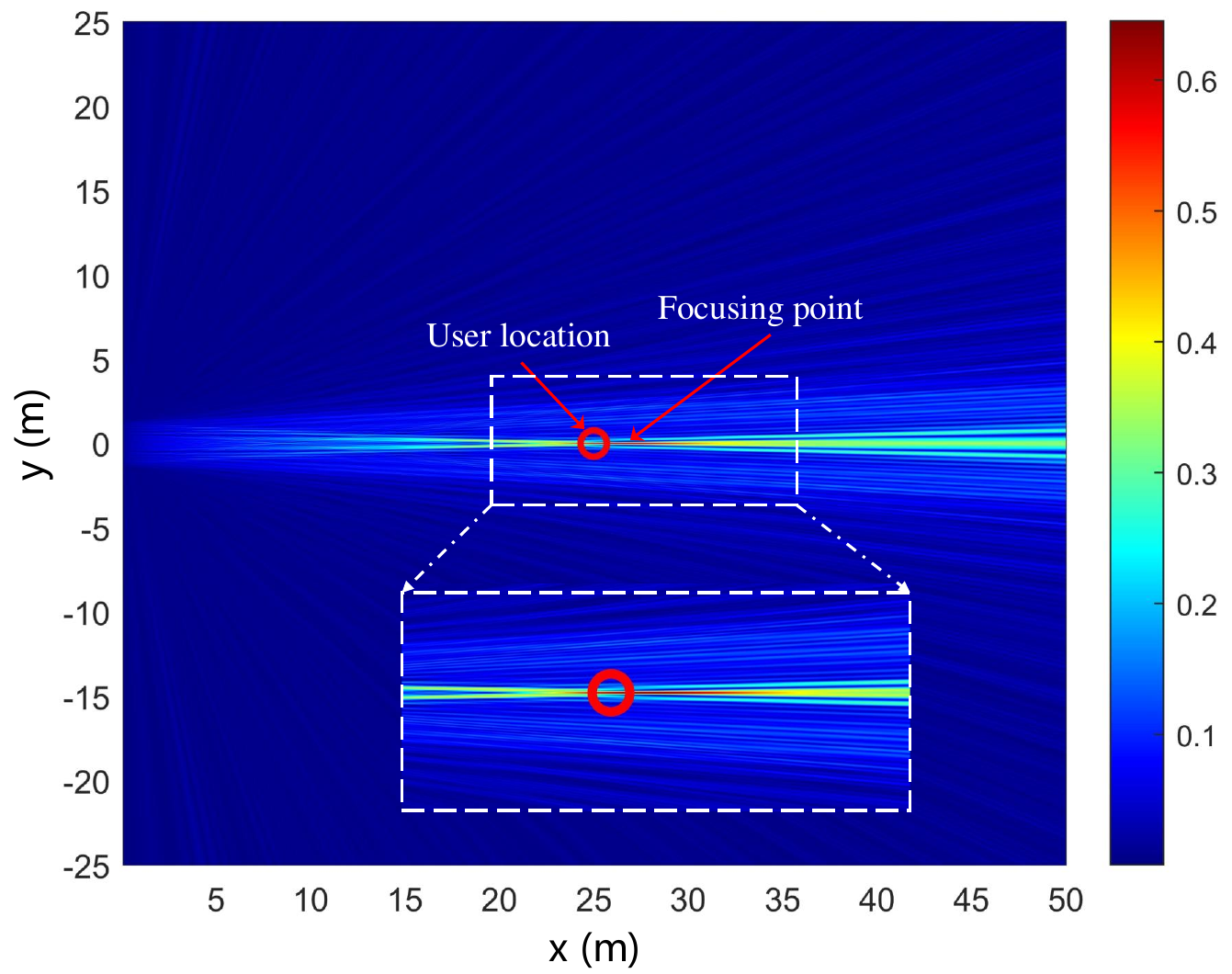} 
    \caption{The overlap case in the near-field with $N=513$, $\theta_{\rm u}=0$, $r_{\rm u}=25$ m under 1-bit PSs.}
    \label{overlapping}
    \vspace{-15pt}
\end{figure}
To shed useful insights, we present below several interesting observations for the corresponding beam pattern. 
In Fig.~\ref{overlapping}, we show the beam pattern when $\theta_{\rm u} = 0$, and $B=1$, for which we have $\theta_k = 0$ for $k \in \mathcal{G}$. This corresponds to the case where all grating lobes occur at the user angle, resulting in the most severe beam overlap.
It is observed that the focused beam location shifts away from the user location (i.e., from $r=25$ m to $r=27$ m). 
However, such beam-focusing offset generally has mile impact on the received beam power even in the most severe case of overlap.
To intuitively explain this phenomenon, we consider two strongest lobes, i.e., the main lobe ($k=1$) and the primary grating lobe ($k=-1$). 
According to \textbf{Proposition 3}, the power of the primary grating lobe increases with $r$. Moreover, the main lobe achieves its maximum beam power at the user location. 
Consequently, after their superposition, the focused beam appears at a range slightly greater than the user range.

The above beam overlap may lead to reduced received power at the user location and increased interference at other locations. 
To mitigate these effects, one can increase the number of antennas or enlarge the PS resolution to enhance the power of main lobe while suppressing that of grating lobes. 

\vspace{-5pt}
\subsection{Understanding the Beam Pattern: An Array-of-subarrays Perspective} 
\label{subarrays}
In this subsection, we provide physical interpretation to intuitively understand the beam patterns in both the far-field and near-field.
In particular, we show that under phase quantization, the entire array can be partitioned into multiple subarrays, for which all antennas in each subarray share the same quantized phase shift. 
Such a virtual array-of-subarrays architecture will be shown to have similar beam pattern with \emph{sparse arrays}. 
This thus leads to the occurrence of grating lobes in both the far-field and near-field cases, while their beam patterns are affected by the corresponding array-of-subarrays structures.

First, we consider the far-field case, for which, based on MRT, the ideal phase at the $n$-th antenna should be set as $\varphi_{{\rm u},n} \!=\!\pi n\sin\theta_{\rm u}$. 
According to the NN phase quantization, when the phase shifts of antennas satisfy $\frac{2\ell\pi}{C_{B}} \leq \varphi_{{\rm u},n} < \frac{2(\ell+1)\pi}{C_{B}},\!\;\!\ell \in \mathbb{Z}$, they are quantized into the same value.
Thus, the indices of antennas with identical quantized phases are given by $\frac{2\ell}{C_{B}\sin\theta_{\rm u}}\!\leq \!n \!< \!\frac{2(\ell+1)}{C_{B}\sin\theta_{\rm u}}, \!\;\!\ell \in \mathbb{Z}$. 
The number of antennas in each subarray can be approximated as $V_{\rm far}=\lceil\frac{2}{C_{B}\sin\theta_{\rm u}}\rceil$, which is almost the same for different subarrays.
As such, the entire antenna array has similar beam pattern with the one of multiple uniform subarrays, as illustrated in Fig. \ref{nfig.3}(a).
Specifically, the beam pattern function of the entire array can be expressed as the Kronecker product of the beam pattern of an LSA (with a sparsity of $S=\frac{2}{C_{B}\sin\theta_{\rm u}}$, formed by the centers of the subarrays) and that of a ULA with identical phase. 
Consequently, the locations of grating lobes in the far-field are primarily attributed to the equivalent LSA, with their amplitudes determined by the number of antennas in each equivalent subarray.

\begin{figure}[!t]
\centering
\includegraphics[width=0.7\columnwidth]{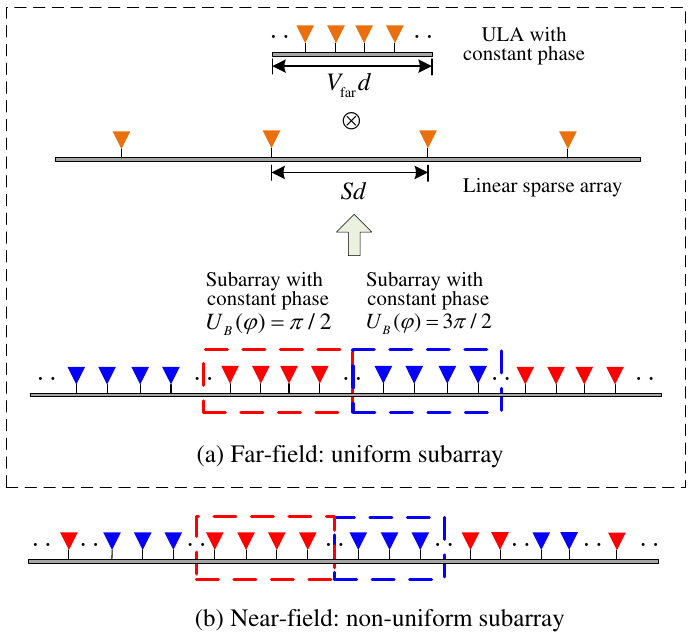}
\caption{The virtual array-of-subarrays architecture under discrete PSs for the far-field and near-field scenarios.} 
\label{nfig.3}
\vspace{-17pt}
\end{figure}
Second, we consider the near-field case. Based on MRT, the ideal phase shift at the $n$-th antenna is non-linear with $n$, i.e., $\varphi_{{\rm u},n} = \frac{2\pi}{\lambda}(nd\sin\theta_{\rm u}-\frac{n^2d^2\cos^2\theta_{\rm u}}{2r_{\rm u}})$.
Similar to the array partitioning method as in the far-field case, it can be shown that the entire array can be divided into multiple subarrays. In particular, the number of antennas in the $\ell$-th subarray can be obtained as
\begin{equation*}
V_{{\rm near}, \ell}(\theta_u, r_u)\!\! =\!\! \bigg\lceil\!\frac{4}{C_{B}\!\!\left(\sqrt{X_\ell(\theta_{\rm u}, r_{\rm u})} \!+ \!\!\sqrt{X_{\ell+1}(\theta_{\rm u}, r_{\rm u})}\right)}\!\bigg\rceil, \ell \in \mathbb{Z},
\end{equation*}
where $X_\ell(\theta_{\rm u}, r_{\rm u})=\sin^2\theta_{\rm u}-\frac{2\lambda\cos^2\theta_{\rm u}}{r_{\rm u}C_B}\ell$.
$V_{{\rm near},\ell}$ monotonically decreases with $r_u$ and equals to $V_{\rm far}$ when $r_{\rm u}\to \infty$.
As such, for the near-field case, the entire array can be equivalently regarded as the one composing of a number of \emph{non-uniform} subarrays, with the same phase shift at the antennas of each subarray, as illustrated in Fig.~\ref{nfig.3}(b).
Again, this equivalent array-of-subarrays architecture leads to the grating lobes, while its near-field beam pattern is generally more difficult to characterize. 
This issue can be effectively tackled by the proposed FSE method to obtain useful insights.

\vspace{-5pt}
\subsection{Multi-user Near-field Communications}
The discrete PSs generally introduce undesirable grating lobes, making the design of multi-user communication systems more challenging. 
In particular, users located at the position of grating lobes may experience severe interference, resulting in degraded communication performance. 
In Section \ref{Numerical Results}, we will present numerical results to showcase the impact of grating lobes on the communication performance. 
To mitigate the impact of such interferences, digital beamforming techniques, such as ZF, MMSE, or optimization-based schemes, can be employed to effectively suppress IUI.

On the other hand, low-resolution PSs offer notable advantages in terms of EE by significantly reducing the system power consumption, which is of paramount importance for XL-arrays of a massive number of antennas.
In particular, EE can be defined as $E_{\mathrm{eff}}\!\!=\!\!\frac{R}{P_{\mathrm{total}}}$ in units of bit/Hz/Joule,
where $R$ is the multi-user sum-rate, and $P_{\mathrm{total}} \!\!=\!\!P_{\mathrm{BB}}\!+\!MP_{\mathrm{RF}}\!+\!MNP_{\mathrm{PS}}\!+\!P_{\mathrm{t}}$ represents total power consumption.
Herein, $P_{\mathrm{BB}}$ is the baseband power consumption, $P_\mathrm{t}$ is the transmit power of the XL-array, and $P_{\mathrm{RF}}$ and $P_{\mathrm{PS}}$ denote the power consumption of a RF chain and PS, respectively. 
It can be observed that the power consumption of both RF chains and PSs in $P_{\mathrm{total}}$ is proportional to the number of users, while the consumption of PSs is proportional to the number of antennas.
Therefore, for scenarios of a large number of users and antennas, the massive number of PSs can incur significant power consumption if high-resolution PSs are used. 
This highlights the practical importance of using low-resolution PSs. 

\section{Numerical Results}
\label{Numerical Results}
In this section, we present numerical results to demonstrate the effects of discrete PSs. 
The system parameters are set as follows. 
The BS is equipped with $N = 513$ antennas, the carrier frequency is set as $f = 60$ GHz, and the reference channel gain is $\beta_0=(\lambda/4\pi)^2=-62~\mathrm{dB}$. Moreover, the (reference) signal-to-noise ratio (SNR) is defined as $\mathrm{SNR}=\frac{P_\mathrm{u}\beta_0 N}{r_{\rm u}^2\sigma^2}$, where $P_\mathrm{u}$ is the transmit power for user~${\rm u}$ and the noise power is $\sigma^2=-70~\mathrm{dBm}$~\cite{zhou2024sparse}.
Without specified otherwise, we set the number of antennas as $N=513$.
All numerical results are averaged over 1000 random realizations of user locations.

In Fig.~\ref{fig.10}, we plot the achievable rates for a two-user setup versus the reference SNR under different resolutions of PSs. 
Two cases of user locations are considered, where both users are located at different angles and their ranges are uniformly distributed in $[20, 60]$ m. The key differences between the two cases is that in Case 1, user 2 is located at the angle of grating lobes of user 1, when the MRT beamforming is employed; while in Case 2, user 2 is not significantly affected by the grating lobes of user 1. 
In order to investigate the effect of grating lobes, we first consider the purely analog beamformer based on MRT, for which individual beams are steered towards the two users, respectively.
For 1-bit PSs, one can observe that the sum-rate in Case 2 is much larger than that in Case 1. 
Therefore, it is necessary to employ digital beamforming (e.g., ZF and MMSE) to further eliminate IUI when low-resolution PSs are employed.
Next, we consider a hybrid beamforming scheme, wherein the analog beamformers are designed based on MRT, followed by MMSE-based digital beamformers to further eliminate IUI. 
First, it is observed that there is a small gap between the sum-rates of Case 1 and Case 2.
Moreover, compared to continuous PSs, there is significant performance degradation when employing 1-bit and 2-bit PSs, while the case with 3-bit PSs achieves similar sum-rate performance with the continuous counterpart.

Fig.~\ref{fig.12} shows the achievable sum-rate versus the number of antennas. 
We consider a system setup where transmit power is $P_{\mathrm{t}}=35~\mathrm{dBm}$ and $M=10$ users are randomly distributed in the angular region $(-\frac{\pi}{2},\frac{\pi} {2})$ and distance region $[20, 60]$ m. 
Digital beamforming is designed based on the optimization method in~\cite{shen2018fractional}. 
It is observed that as the number of antennas increases, all schemes achieve increasing sum-rate, while the growth rate is gradually diminishing. 
This is because, when the number of antennas is small, increasing the antenna number makes the channel model gradually shift from far-field to the near-field, thereby achieving the beam-focusing gain.
However, the sum-rate performance gap between low-resolution and continuous PSs gradually diminishes, when the number of antennas further increases. 
This is because the focused regions of the main lobe and Type-I grating lobes become smaller, and the power of Type-II grating lobes decreases progressively. 
Moreover, when the number of antennas is sufficiently large, hybrid beamforming (e.g., MRT+MMSE) achieves very close sum-rate performance to that of digital beamforming due to enhanced channel orthogonality.

\begin{figure}[!t]
  \centering
\includegraphics[width=0.7\columnwidth]{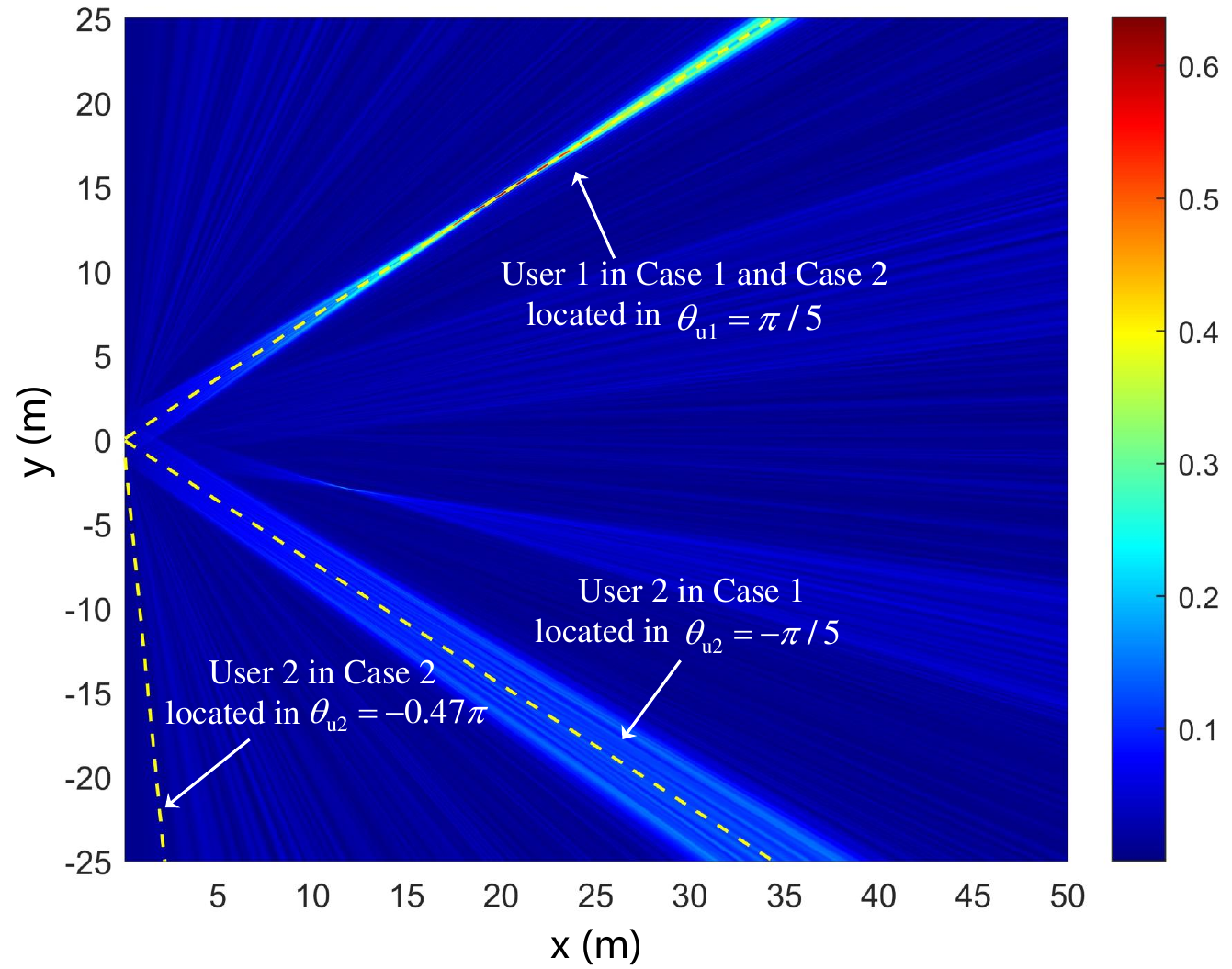} 
\caption{Locations of users for Case 1 and Case 2. The near-field beam pattern refers to the MRT beamforming for user 1 at ($\theta_{\rm u}=\pi/5, r_{\rm u}=25$ m) under 1-bit PSs.}
\label{distribution}
\vspace{-17pt}
\end{figure}

In Fig.~\ref{fig.EE}, we plot EE versus the transmit power for the hybrid beamforming scheme based on MRT and MMSE. We consider $M = 10$ users, and the user locations follow the same distribution as that in Fig.~\ref{fig.12}.
Moreover, the baseband and RF power consumption are set as $P_{\mathrm{BB}}\!=\!200~\mathrm{mW}$ and  $P_{\mathrm{RF}}\!=\!240~\mathrm{mW}$, respectively~\cite{xie2019power}, and the power consumption of 1-bit, 2-bit, 3-bit and 4-bit PSs are set as $5~\mathrm{mW}$, $10~\mathrm{mW}$, $15~\mathrm{mW}$ and $45~\mathrm{mW}$, respectively~\cite{mendez2016hybrid}.
For example, in this setup, compared to the 3-bit PS case, employing 1-bit PSs can reduce the total power consumption by $51.2~\mathrm{W}$, which accounts to $62\%$ of total power consumption, when transmit power is $P_{\mathrm{t}}=35~\mathrm{dBm}$.
As such, it is shown that the EE of the low-resolution PSs is much higher than that of high-resolution PSs, since they can significantly reduce the power consumption while maintaining high beamforming gain.

Last, in Fig. \ref{fig.EE_usernum}, we show the effect of number of users on EE. 
The main observations are summarized as follows.
First, it is observed that the EE of 4-bit PSs decreases with the number of users.
This is due to the fact that 4-bit PSs result in much higher power consumption as the number of users grows, thereby reducing EE.
In contrast, for lower-bit cases, when the number of users increases, the system EE first increases and then gradually decreases.
This is expected, since the lower-resolution PSs incur slight increasing power consumption, while the sum-rate performance are greatly improved when the number of users increases.
However, when the number of users is sufficiently large, the heavy power consumption eventually results in decreased system EE.
\begin{figure}[!t]
\vspace{-10pt}
 \centering
\includegraphics[width=0.7\columnwidth]{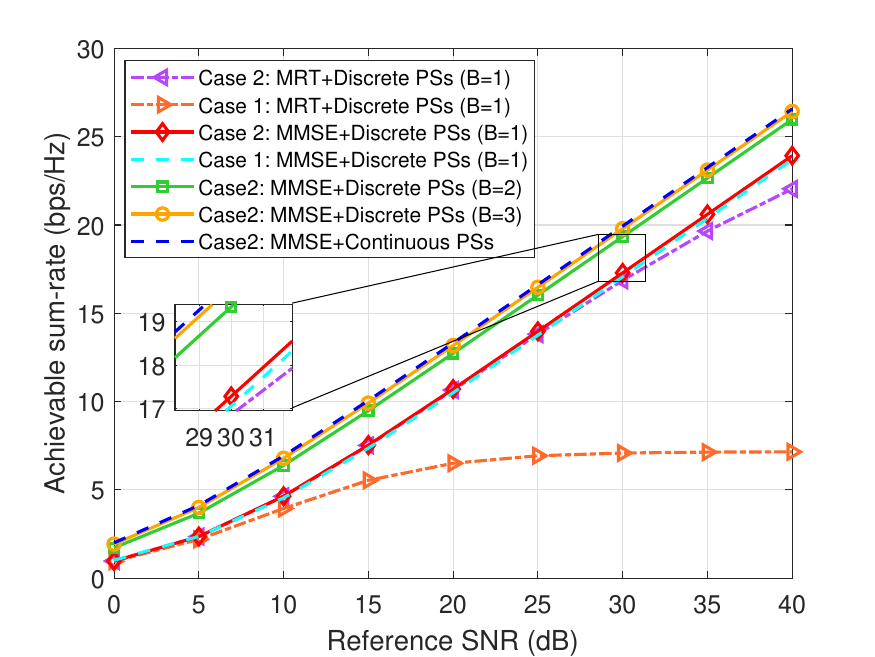}
\caption{Achievable sum-rate versus reference SNR.}
\label{fig.10}
\vspace{-14pt}
\end{figure}
\begin{figure}[!t]
\centering
\includegraphics[width=0.7\columnwidth]{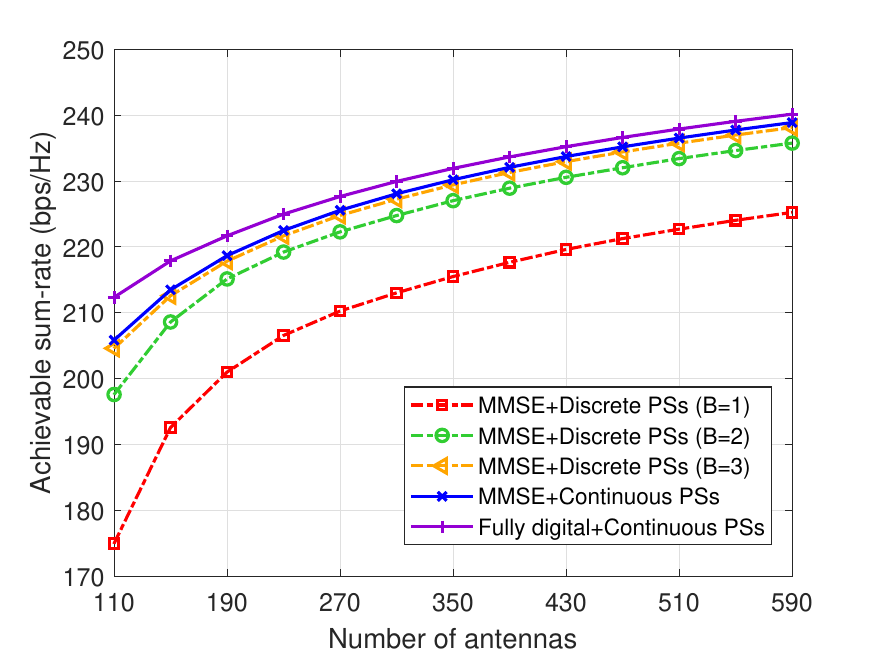}
\caption{Achievable sum-rate versus number of antennas.}
\label{fig.12}
\vspace{-14pt}
\end{figure}

\section{Conclusions}
In this paper, we proposed a new FSE method to characterize the beam pattern under discrete PSs in the near-field.
It was shown that discrete PSs introduce undesired grating lobes, while the main lobe still exhibits the beam focusing effect, with its beam power increasing with PS resolution. 
Moreover, we showed that the grating lobes can be divided into two types, which are characterized by the beam-focusing and beam-steering effects, respectively. 
It was revealed that, the sum-rate performance of near-field is superior over the far-field scenario, due to the lower grating lobe power and smaller spatial region covered by grating lobes.
Numerical results showed that grating lobes generally reduce communication rate, which can be mitigated by MMSE beamforming in digital ends. 
Moreover, 3-bit PSs can achieve similar beam patterns and rate performance with the continuous PS counterpart, while it attains much higher EE.

\section*{Appendix A: Proof of Lemma \ref{fourier coefficients}}
\label{proof of fourier coefficients}
According to NN quantization, the quantized phase takes the midpoint value of the quantization interval. Thus, we have 
\begin{equation*}
\label{neq18}
\!Q_{B}(\varphi)=e^{\jmath\frac{(2\ell+1)\pi}{C_{B}}}, \;\! \text{when } \frac{2\ell\pi}{C_{B}} \!\leq\! \varphi \!<\! \frac{2(\ell+1)\pi}{C_{B}}, \ell \in \mathbb{Z}.
\end{equation*}
Then, the Fourier coefficients in \eqref{eq9} can be obtained as
\begin{equation}
\label{neq19}
a_k^{(B)}=\frac{1}{2\pi}\sum_{\ell=0}^{C_{B}-1}\int_{\frac{2\ell\pi}{C_{B}}}^{\frac{2(\ell+1)\pi}{C_{B}}}e^{\jmath\frac{(2{\ell}+1)\pi}{C_{B}}}e^{-\jmath k\varphi}d\varphi.
\end{equation}
When $k=0$, \eqref{neq19} can be rewritten as
\begin{align*}
a_0^{(B)} &\!=\! \frac{1}{2\pi}\!\!\sum_{\ell=0}^{C_{B}-1}\!\int_{\frac{2\ell\pi}{C_{B}}}^{\frac{2(\ell+1)\pi}{C_{B}}}\!\!\!e^{\jmath\frac{(2{\ell}+1)\pi}{C_{B}}}\!d\varphi \!=\! \frac{1}{C_{B}}\!\!\sum_{\ell=0}^{C_{B}-1}\!e^{\jmath\frac{(2{\ell}+1)\pi}{C_{B}}} \\
&=\! \frac{e^{\jmath\frac{\pi}{C_{B}}}}{C_{B}}\frac{1-e^{\jmath2\pi}}{1-e^{\jmath\frac{2\pi}{C_{B}}} }= 0.
\end{align*}
When $k\in \mathcal{T}$, the expression of $a_k^{(B)}$ is given by
\begin{equation}
a_k^{(B)}\!\!=\!\!\frac{\sin(\frac{k\pi}{C_{B}})}{k\pi}e^{\jmath\frac{(1-k)\pi}{C_{B}}}\!\!\sum_{\ell=0}^{C_{B}-1}\!e^{\jmath\frac{2(1-k)\pi}{C_{B}}\ell}\! \stackrel{(a_1)}{=}\!\!\frac{C_{B}}{k\pi}\!\sin(\frac{\pi}{C_{B}}),\! \label{eq41}
\end{equation}
where $(a_1)$ is obtained by substituting $k=1-pC_{B}$ into \eqref{eq41}.
When $k\neq0$ and $k\notin \mathcal{T}$, \eqref{neq19} is given by
\begin{align*}
a_k^{(B)}&=\frac{\sin(\frac{k\pi}{C_{B}})}{k\pi}e^{\jmath\frac{(1-k)\pi}{C_{B}}}\sum_{\ell=0}^{C_{B}-1}e^{\jmath\frac{2(1-k)\pi}{C_{B}}\ell}\\
&=\frac{\sin(\frac{k\pi}{C_{B}})}{k\pi}e^{\jmath\frac{(1-k)\pi}{C_{B}}}\frac{1-e^{\jmath2(1-k)\pi}}{1-e^{\jmath\frac{2(1-k)\pi}{C_{B}}}}=0.
\end{align*}
Combining the above three cases leads to the desired result.

\begin{figure}[!t]
\vspace{-10pt}
\centering
\includegraphics[width=0.7\columnwidth]{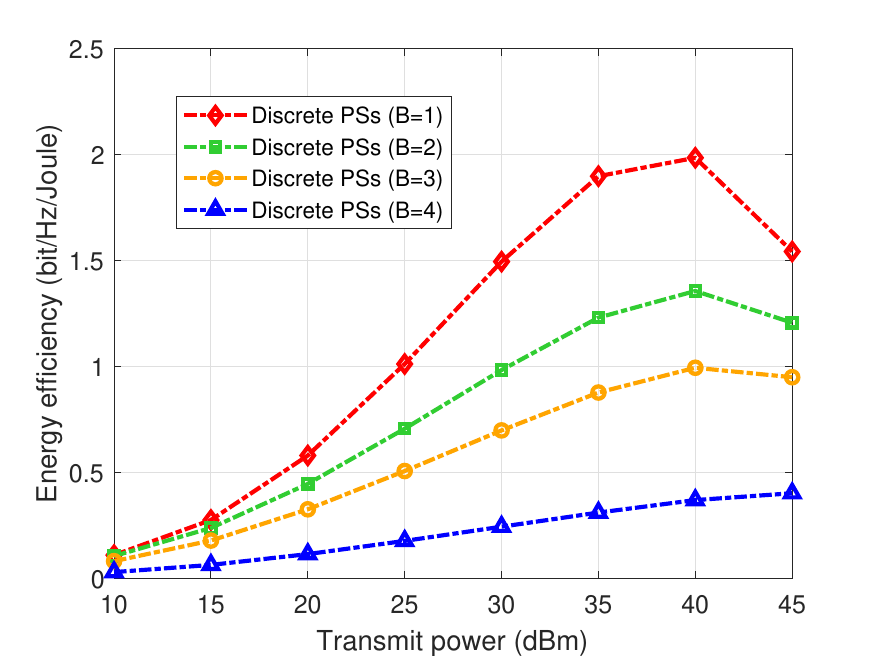}
\caption{EE versus BS transmit power.}
\label{fig.EE}
\vspace{-15pt}
\end{figure}
\begin{figure}[!t]
\centering
\includegraphics[width=0.7\columnwidth]{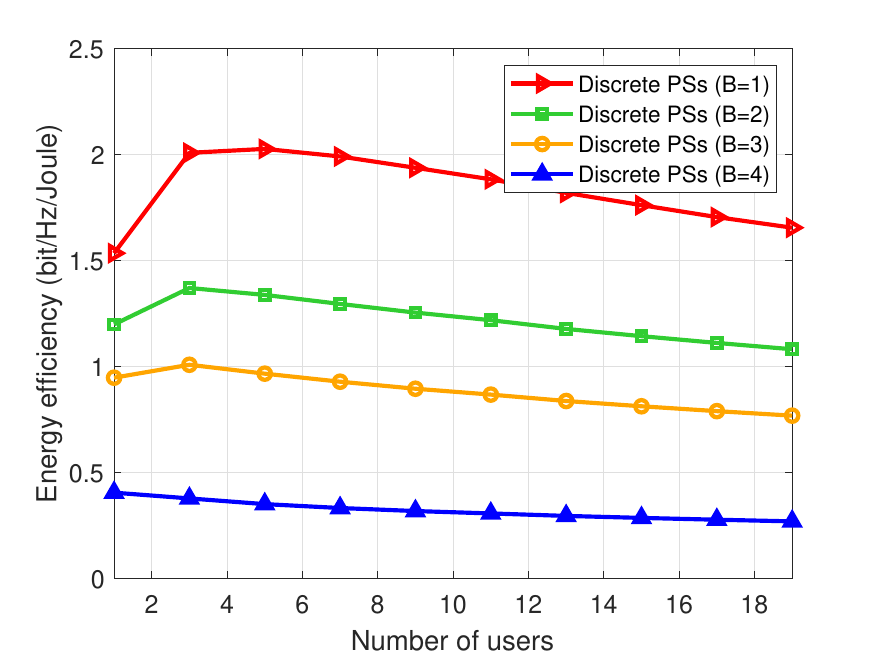}
\caption{EE versus number of users.}
\label{fig.EE_usernum}
\vspace{-16pt}
\end{figure}
\section*{Appendix B: Proof of Property \ref{principle term} }
\rm First, the property of periodicity can be obtained directly from \eqref{neq14}. 
Second, according to the property about the sparsity and considering periodicity, we can conclude that $|f_k^{(B)}(\theta,r;\theta_{\rm u},r_{\rm u})|$ has large values around the angle of $\theta_k$, which is obtained from
$
\label{nnn24}
\Delta_k=k\sin\theta_{\mathrm{u}}-\sin\theta=2m,m\in\mathbb{Z}.
$
\rm 
Then, we prove \eqref{eq23} by contradiction. Assume that there are two distinct values $\Theta_1, \Theta_2 \in \left(-\frac{\pi}{2}, \frac{\pi}{2}\right)$ that both satisfy the condition of $\Delta_k$. That is, $k\sin\theta_{\rm u} - \sin\Theta_1 = 2m,  m \in \mathbb{Z}$, and $k\sin\theta_{\rm u} \!-\! \sin\Theta_2 \!=\! 2m^{\prime},  m^{\prime} \in \mathbb{Z}$. Combining these leads to
$
\sin\Theta_2 - \sin\Theta_1 = 2(m - m^{\prime}),\! \;\! m, m^{\prime} \in \mathbb{Z}.
$
Given $\Theta_1, \Theta_2 \!\!\in\!\! \left(-\frac{\pi}{2}, \frac{\pi}{2}\right)$, it can be easily shown that $\sin\Theta_1 - \sin\Theta_2 \in (-2, 2)$. 
On the other hand, for $m, m^{\prime} \in \mathbb{Z}$, $\sin\Theta_1 - \sin\Theta_2=2(m - m^{\prime})$ is an even integer. Combining the above, we conclude that $m = m^{\prime}$, which is equivalent to $\sin\Theta_1 = \sin\Theta_2$, implying that $\Theta_1 = \Theta_2$. This contradicts the assumption and thus leads to the result in \eqref{eq23}.
Moreover, when $k>0$, we can also conclude that $|f_k^{(B)}(\theta,r;\theta_{\rm u},r_{\rm u})|$ achieves its maximum value $|a_k^{(B)}|$ under the following conditions
\begin{align}
&\!\!\!\!\!\!{\text{Condition 1:}}~~  \Delta_k=k\sin\theta_{\mathrm{u}}-\sin\theta=2m,m\in\mathbb{Z},\label{eq24}\\
&\!\!\!\!\!\!{\text{Condition 2:}}~~\Phi_k=-k\frac{\cos^{2}\theta_{\mathrm{u}}}{r_{\mathrm{u}}}+\frac{\cos^{2}\theta}{r}=0.\label{eq25}
\end{align}
\eqref{eq24} is the same as \eqref{nnn24}, thus \eqref{Eq:MVAngle} can be obtained from \eqref{eq24} similarly. Moreover, \eqref{Eq:MVRange} can be obtained from \eqref{eq25} directly.
On the other hand, when $k\leq 0$, depending on the range of $0\leq-k\frac{\cos^2\theta_{\rm u}}{r_{\rm u}}\!<\!\Phi_k\!\leq\!\infty$, we can conclude that $|f_k^{(B)}(\theta,r;\theta_{\rm u},r_{\rm u})|$ cannot reach a maximum at a specific location. Instead, it has large values around $\theta_k$ due to the angular sparsity.

\section*{Appendix C: Proof of the expressions of beam-width in proposition \ref{Main lobe} and \ref{Type-I grating lobes}}\label{App:PropoBeam}
The beam-widths of main lobe and Type-I grating lobes are defined at the corresponding distance ring, i.e., $\{(r, \theta)|\frac{\cos^2\theta}{r}=k\frac{\cos^2\theta_{\rm u}}{r_{\rm u}} \}$ with $ k \in (\mathcal{G}_1 \cup \{1\})$, where we have $\Phi_k=0$. Then, the beam function of each lobe can be rewritten as
    \begin{align}
    \left|f_k^{(B)}(\theta,r;\theta_{\rm u},r_{\rm u})\right|&=\left|\frac{1}{N}a_{k}^{(B)}\sum_{n=-\tilde{N}}^{\tilde{N}}e^{\jmath \pi n(k\sin\theta_{\rm u}-\sin\theta)}\right|
    \nonumber
    \\&=\left|a_{k}^{(B)}\Xi_{N}(\pi(k\sin\theta_{\rm u}-\sin\theta))\right|,
    \label{nneq23}
    \end{align}
    where $\!\Xi_{N}(\alpha)\!=\!\frac{\sin\frac{N}{2}\alpha}{N\sin\frac{1}{2}\alpha}$ is the Dirichlet sinc function.
    When the array aperture is large, the beam pattern in \eqref{nneq23} can be approximated as~\cite{zhou2024sparse}
    \begin{equation*}
    \left|f_k^{(B)}(\theta,r;\theta_{\rm u},r_{\rm u})\right|=\left|a_{k}^{(B)}\mathrm{sinc}(\frac{N(k\sin\theta_{\rm u}-\sin\theta)}{2})\right|,
    \end{equation*}
    where $\mathrm{sinc}(x) = \frac{\sin(\pi x)}{\pi x}$ denotes the sinc function. To characterize the 3-dB beam-width, we numerically set $\frac{N(k\sin\theta_{\rm u}-\sin\theta)}{2}=\pm0.44$ and thus obtain $ \mathrm{BW}_k^{(B)} = \frac{1.76}{N}$.

\section*{Appendix D: Proof of the expressions of beam-depth in proposition \ref{Main lobe} and \ref{Type-I grating lobes}}
\label{proof of beam-depth}
The function $\left|f_k^{(B)}(\theta,r;\theta_{\rm u},r_{\rm u})\right|$ in the range domain at an angle $\theta_k$ is given as,
\begin{equation*}
\left|\frac1N a_k^{(B)}\!\!\!\sum_{n=-\tilde{N}}^{\tilde{N}}\!\!\!e^{\jmath \frac{\pi\lambda}4n^2(-k\frac{\cos^2\theta_{\rm u}}{r_{\rm u}}+\frac{\cos^2\theta_k}{r})}\right|\!=\!\left| a_k^{(B)} \right| \left| L(x) \right|,
\end{equation*}
where $x = \frac{\lambda}{4}(-k\frac{\cos^2\theta_{\rm u}}{r_{\rm u}}+\frac{\cos^2\theta_k}{r})$ and the function $L(x)$ is
\begin{equation*}
\begin{aligned}L(x)&=\frac{1}N\sum_{n=-\tilde{N}}^{\tilde{N}}e^{\jmath \pi n^2x}&\approx\frac1N\int_{-N/2}^{N/2}e^{\jmath \pi n^2x}\mathrm{d}n.\end{aligned}
\label{eq29}
\end{equation*}
According to the proof of~\cite{cui2022channel}, $|L(x)|$ can be written as
$
|L(x)|\approx\left|\frac{C(\beta)+\jmath S(\beta)}\beta \right|,
$
where $C(\beta)=\int_{0}^{\beta}\cos(\frac\pi2t^2)\mathrm{d}t$ and $S(\beta)=\int_{0}^{\beta}\sin(\frac\pi2t^2)\mathrm{d}t$ are Fresnel functions~\cite{cui2022channel}. Moreover, the $\beta$ is given by
\begin{equation}
\beta=\frac{\sqrt{2x}N}{2}=\sqrt{\frac{N^2\lambda}{8}\left|-k\frac{\cos^2\theta_{\rm u}}{r_{\rm u}}+\frac{\cos^2\theta_k}{r}\right|}.
\label{eq48}
\end{equation}  
Based on numerical results, $\left|\frac{C(\beta)+\jmath S(\beta)}\beta\right|\geq\frac{\sqrt{2}}2\text{ for }\beta\leq1.31$. Therefore, by defining $\eta_{3\mathrm{dB}} = 1.31$, we approximately have $|L(x)|\geq\frac{\sqrt{2}}2$ when 
$
\sqrt{\frac{N^2\lambda}{8}\left|-k\frac{\cos^2\theta_{\rm u}}{r_{\rm u}}+\frac{\cos^2\theta_k}{r}\right|} \leq \eta_{3\mathrm{dB}},
$
which is equivalent to 
\begin{equation}
\!\!\!\!\left(\!\frac{k\cos^2\theta_{\rm u}}{r_{\rm u}\!\cos^2\theta_k}\!\!-\!\!\frac{8\eta_{3\mathrm{dB}}^2}{N^2\lambda\cos^2\theta_k}\!\!\right)^+\!\!\!\!\!\leq\!\frac1{r}\!\!\leq\!\!\left(\!\frac{k\cos^2\theta_{\rm u}}{r_{\rm u}\!\cos^2\theta_k}\!\!+\!\!\frac{8\eta_{3\mathrm{dB}}^2}{N^2\lambda\cos^2\theta_k}\!\!\right)^+\!\!\!\!.
\label{eq31}
\end{equation}

As for $k>0$, if $r_{\rm u}<kr_\mathrm{DF}$, where $r_\mathrm{DF}\triangleq \frac{N^2\lambda \cos^2\theta_{\rm u}}{8\eta^2_{\mathrm{3dB}}}$,~\eqref{eq31} can be deduced to
\begin{equation*}
\!\!\!\!\frac{k\cos^2\theta_{\rm u}}{r_{\rm u}\cos^2\theta_k}\!-\!\frac{8\eta_{3\mathrm{dB}}^2}{N^2\lambda\cos^2\theta_k}\!\!\leq\!\!\frac1{r}\!\!\leq\!\!\frac{k\cos^2\theta_{\rm u}}{r_{\rm u}\cos^2\theta_k}\!+\!\frac{8\eta_{3\mathrm{dB}}^2}{N^2\lambda\cos^2\theta_k}.
\label{eq32}
\end{equation*}
Therefore, the range of $r$ is given by
\begin{equation*}
\frac{\cos^2\theta_k}{\cos^2\theta_{\rm u}}\frac{r_{\rm u}r_\mathrm{DF}}{kr_\mathrm{DF}+r_{\rm u}}\leq r\leq\frac{\cos^2\theta_k}{\cos^2\theta_{\rm u}}\frac{r_{\rm u}r_\mathrm{DF}}{kr_\mathrm{DF}-r_{\rm u}}.
\label{eq33}
\end{equation*}
As such, the beam-depth for the dominant lobes represented by $\left|f_k^{(B)}(\theta,r;\theta_{\rm u},r_{\rm u})\right|,~k>0$ is given by
\begin{equation*}
\begin{aligned}
\mathrm{BD}_k^{(B)}&=\frac{\cos^2\theta_k}{\cos^2\theta_{\rm u}}\frac{r_{\rm u}r_\mathrm{DF}}{kr_\mathrm{DF}-r_{\rm u}}-\frac{\cos^2\theta_k}{\cos^2\theta_{\rm u}}\frac{r_{\rm u}r_\mathrm{DF}}{kr_\mathrm{DF}+r_{\rm u}} \\
&=\frac{\cos^2\theta_k}{\cos^2\theta_{\rm u}}\frac{2r_{\rm u}^2r_\mathrm{DF}}{k^2r_\mathrm{DF}^2-r_{\rm u}^2}.
\end{aligned}
\label{eq34}
\end{equation*}
If $r_{\rm u}\geq kr_\mathrm{DF}$, the range of $r$ is given by
$
\frac{\cos^2\theta_k}{\cos^2\theta_{\rm u}}\frac{r_{\rm u}r_\mathrm{DF}}{kr_\mathrm{DF}+r_{\rm u}}\leq r\leq \infty.
$
As such, $\mathrm{BD}_k^{(B)}=\infty$, thus completing the proof.

\begin{figure}[!t]
\vspace{-10pt}
\centering
\includegraphics[width=0.75\columnwidth]{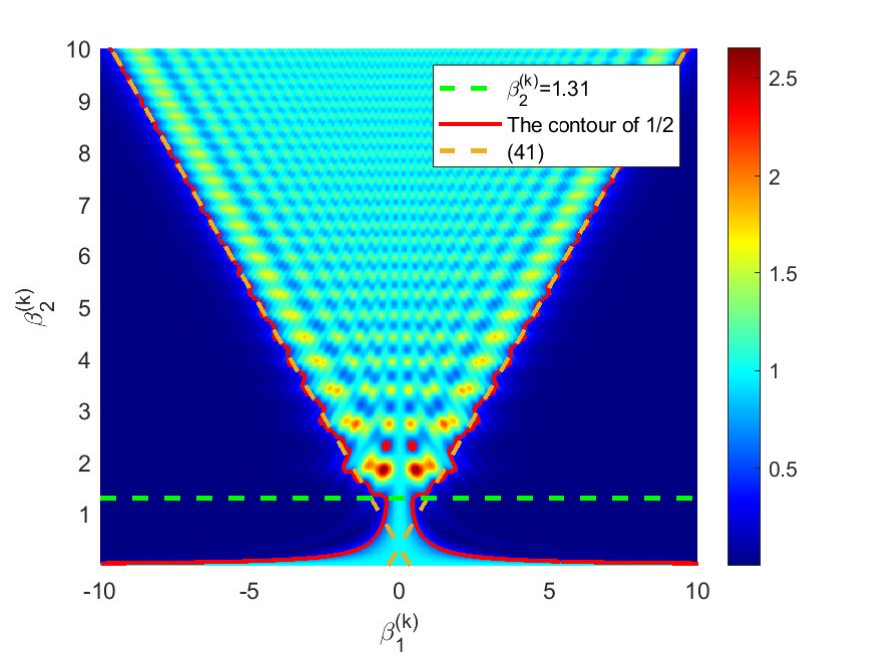}
\caption{The power ratio $\mathcal{R}(\Delta_k,\Phi_{k,0})$ in \eqref{ratio}.} 
\label{fig.ratio}
\vspace{-15pt}
\end{figure}
\section*{Appendix E: Proof of proposition \ref{Type-II grating lobes}}
The beam power of Type-II grating lobe at $\theta_k$ is given by $|f_k^{(B)}(\theta_k,r;\theta_{\rm u},r_{\rm u})|^2=\left| a_k^{(B)} \right|^2 \left| L(x) \right|^2$.
Since we have $k<0$ for Type-II grating lobes, $\beta$ in~\eqref{eq48} decreases with $r$ and $r_{\rm u}$ and increases with $N$. Moreover, $|L(x)|$ increases with $\beta$ (according to the proof of~\cite{cui2022channel}). Thus, $|L(x)|$ increases with $r$ and $r_{\rm u}$, and decreases with $N$. In addition, $|a_k^{(B)}|$ decreases with $B$. Thus, we can conclude that the beam power of Type-II grating lobe increases with $r$, while it decreases with $N$ and $B$.
Therefore, given $r_{\rm u}$, $N$ and $B$, the beam power is upper-bounded by $|f_k^{(B)}(\theta,r;\theta_{\rm u},r_{\rm u})|^2_{r\to\infty}=\left| a_k^{(B)} \right|^2 \left| \frac{C(\beta_{\infty})+\jmath S(\beta_{\infty})}{\beta_{\infty}} \right|^2$, where $\beta_{\infty}=\sqrt{\frac{N^2\lambda}{8}\left|-k\frac{\cos^2\theta_{\rm u}}{r_{\rm u}}\right|}$.

The surrogate beam-width of Type-II grating lobe at $\theta_k$ and $r_0$ is defined on the ring $\Phi_{k,0} = -k\frac{\cos^2\theta_{\rm u}}{r_{\rm u}}+\frac{\cos^2\theta_k}{r_0}$, where the power drops to half of its value at $\theta_k$ on this ring. 
To address this problem, we first define the power ratio as \cite{wu2024near}
\begin{equation*}
\!\mathcal{R}(\Delta_k,\Phi_{k,0})\!\triangleq\!\frac{|f_k^{(B)}(\theta,r;\theta_{\rm u},r_{\rm u})|^2}{|f_k^{(B)}(\theta_k,r_0;\theta_{\rm u},r_{\rm u})|^2}\!=\!\frac{|H(\Delta_k, \Phi_{k,0})|^2}{|H(0,\Phi_{k,0})|^2}.
\end{equation*}
Then, we need to solve the equation of $\mathcal{R}(\Delta_k,\Phi_{k,0})=\frac{1}{2}$ to obtain the solution $\Delta_k^*$, for which the surrogate beam-width is given by $\mathcal{A}^{(B)}_k(r_0)=2\Delta_k^*$.
According to the proof of Lemma 2 in \cite{zhou2024sparse}, $|H(\Delta, \Phi)|$ can be written as
$
|H(\Delta, \Phi)| \approx \left|\frac{\widehat{C}(\beta_{1},\beta_{2}) +\jmath \widehat{S}(\beta_{1},\beta_{2})}{2\beta_2} \right|,
$
where $\widehat{C}(\beta_{1},\beta_{2}) \triangleq C(\beta_1+\beta_2)-C(\beta_1-\beta_2)$ and $\widehat{S}(\beta_{1},\beta_{2})\triangleq S(\beta_1+\beta_2)-S(\beta_1-\beta_2)$. $\beta_1$ and $\beta_2$ are defined as $\beta_1=\frac{\Delta}{\sqrt{d\left|\Phi\right|}},\beta_2=\frac{N}{2}\sqrt{d\left|\Phi\right|}$. Then, the power ratio is rewritten as
\begin{equation}
\!\mathcal{R}(\Delta_k,\Phi_{k,0}) \!=\! \frac{[\widehat{C}(\beta_1^{(k)},\beta_2^{(k)})]^{2}+[\widehat{S}(\beta_1^{(k)},\beta_2^{(k)})]^{2}}{4[C^{2}(\beta_{2}^{(k)})+S^{2}(\beta_{2}^{(k)})]},
\label{ratio}
\end{equation}
where $\beta_1^{(k)}=\frac{\Delta_k}{\sqrt{d\left|\Phi_{k,0}\right|}},\beta_2^{(k)}=\frac{N}{2}\sqrt{d\left|\Phi_{k,0}\right|}$.

To solve $\mathcal{R}(\Delta_k,\Phi_{k,0})=\frac{1}{2}$, we first give some useful approximations. When $\beta$ is small, we have $C(\beta)\approx\beta$ and $S(\beta)\approx0$.
Moreover, when $\beta>1$, we have the following approximations
\begin{align*}
C(\beta)\approx \tilde{C}(\beta)\triangleq\frac{1}{2}+\frac{\sin\left(\frac{\pi}{2}\beta^2\right)}{\pi\beta}, \\
S(\beta)\approx \tilde{S}(\beta)\triangleq\frac{1}{2}-\frac{\cos\left(\frac{\pi}{2}\beta^2\right)}{\pi\beta}.
\end{align*}

Note that, in the effective near-field region, we have $r_{\rm u}<r_\mathrm{DF}$. Thus, $\beta_2^{(k)}=\frac{N}{2}\sqrt{d\left|\Phi_{k,0}\right|}>\frac{N}{2}\sqrt{d\left|k\frac{\cos^2\theta_{\rm u}}{r_{\rm u}}\right|}>\frac{N}{2}\sqrt{d\left|\frac{8\eta_{3\mathrm{dB}}^2}{N^2\lambda}\right|}>\eta_{3\mathrm{dB}}=1.31>1$. Without loss of generality, we can assume $\beta_1^{(k)}>0$, thus we have $\beta_1^{(k)}+\beta_2^{(k)}>1$. 
Moreover, as presented by the contour of $\mathcal{R}(\Delta_k,\Phi_{k,0})=\frac{1}{2}$ in Fig.~\ref{fig.ratio}, we can conclude that $\delta = \beta_1^{(k)}-\beta_2^{(k)}\approx0$. 
Thus, we can use the approximations above to simplify \eqref{ratio}. After neglecting the higher-order oscillatory terms, we obtain $\delta\approx\frac12-\frac12\sqrt{3+\frac{8}{\pi^2(\beta_2^{(k)})^2}}\approx\frac{1-\sqrt{3}}{2}$. Thus, we have
\begin{equation}
\beta_1^{(k)}=\beta_2^{(k)}+\frac{1-\sqrt{3}}{2},
\label{eq61}
\end{equation}
which is accurate for $\beta_2^{(k)}\!\!>\!\!1.31$, as shown in Fig. \ref{fig.ratio}. As such, the surrogate beam-width is given by $\mathcal{A}^{(B)}_k(r_0)=2\Delta_k^*=Nd|\Phi_{k,0}|+(1-\sqrt{3})\sqrt{d|\Phi_{k,0}|}$, from which we can conclude that $\mathcal{A}^{(B)}_k(r_0)$ increases with $N$ and $|\Phi_{k,0}|$, which means that $\mathcal{A}^{(B)}_k(r_0)$ decreases with $r_0$, thus completing the proof.

\bibliographystyle{IEEEtran}

\bibliography{IEEEabrv}
\end{document}